\documentclass[12pt,oneside,english]{article}
\usepackage{amsthm, amsmath, amsfonts, amsxtra, amssymb, euscript, mathrsfs, MnSymbol, verbatim, enumerate, multicol, multirow, color,babel, geometry, tikz,tikz-cd, tikz-3dplot, tkz-graph, array, hepth, enumitem, hyperref, thm-restate, thmtools, datetime, graphicx, tensor, braket, slashed, adjustbox, mathtools, bbding, esint, pgfplots, ytableau, float, rotating, mathdots, savesym, cite, wasysym, amscd, pifont, setspace, wrapfig, picture, subfigure, tabularx}
\usepackage[numbers,sort&compress]{natbib}
\usepackage[normalem]{ulem} 
\usepackage[all]{xy}
\numberwithin{equation}{section}
\numberwithin{figure}{section}
\usetikzlibrary{arrows,positioning,decorations.pathmorphing,   decorations.markings, matrix, patterns}
\hypersetup{colorlinks=true}
\hypersetup{linkcolor=black}
\hypersetup{citecolor=black}
\hypersetup{urlcolor=black}

\usepackage{youngtab}
\usepackage{amssymb}
\usetikzlibrary{arrows}
\usetikzlibrary{decorations.markings}
\usetikzlibrary{decorations.pathreplacing,matrix,snakes}
\tikzset{
  big arrow/.style={
    decoration={markings,mark=at position 1 with {\arrow[scale=2,#1]{>}}},
    postaction={decorate},
    shorten >=0.4pt},
  big arrow/.default=black}

\theoremstyle{plain}
\newtheorem*{thm*}{Theorem}
\newtheorem{thm}{Theorem}[section]

\newtheorem{lem}[thm]{Lemma}

\theoremstyle{definition}
\newtheorem{defn}[thm]{Definition}
\newtheorem*{defn*}{Definition}

\newtheorem{rem}[thm]{Remark}

\newtheorem*{thm:otherdirection}{Theorem \ref{thm:otherdirection}}
\newtheorem*{thm:maintheorem}{Theorem \ref{thm:maintheorem}}

\tikzset{->-/.style={decoration={
  markings,
  mark=at position #1 with {\arrow[scale=2.5]{>}}},postaction={decorate}}}

\makeatother
\newcommand{\calo}{\mathcal{O}}
\newcommand{\cala}{\mathcal{A}}
\newcommand{\calp}{\mathcal{P}}
\newcommand{\cals}{\mathcal{S}}
\newcommand{\calh}{\mathcal{H}}
\newcommand{\calb}{\mathcal{B}}
\newcommand{\calu}{\mathcal{U}}
\newcommand{\calv}{\mathcal{V}}

\begin{document}

\date{}
\institution{HarvardPhys}{\centerline{${}^{1}$Department of Physics, Harvard University, Cambridge, MA, USA}}
\title{Holographic Relative Entropy in Infinite-dimensional Hilbert Spaces}
\authors{Monica Jinwoo Kang\worksat{\HarvardPhys}\footnote{e-mail: {\tt monica@caltech.edu}} and David K. Kolchmeyer\worksat{\HarvardPhys}\footnote{e-mail: {\tt dkolchmeyer@g.harvard.edu}}}
\abstract{We reformulate entanglement wedge reconstruction in the language of operator-algebra quantum error correction with infinite-dimensional physical and code Hilbert spaces. Von Neumann algebras are used to characterize observables in a boundary subregion and its entanglement wedge. Assuming that the infinite-dimensional von Neumann algebras associated with an entanglement wedge and its complement may both be reconstructed in their corresponding boundary subregions, we prove that the relative entropies measured with respect to the bulk and boundary observables are equal. We also prove the converse: when the relative entropies measured in an entanglement wedge and its complement equal the relative entropies measured in their respective boundary subregions, entanglement wedge reconstruction is possible. Along the way, we show that the bulk and boundary modular operators act on the code subspace in the same way. For holographic theories with a well-defined entanglement wedge, this result provides a well-defined notion of holographic relative entropy.}

\maketitle
\tableofcontents
\newpage

\section{Introduction}

Entanglement entropy has many applications in quantum field theory, ranging from the study of renormalization group flows \cite{CasiniHuerta2015, Casini2006} to confinement \cite{Klebanov2007} to topological orders \cite{Wen2006,KitaevPreskill}. With the discovery of the Ryu--Takayangi formula \cite{RT2006}, entanglement entropy has been especially useful in studying holographic quantum field theories. 
For holographic theories, it is important to understand the emergent low-energy bulk physics in $d$-dimensions from the conformal field theory in  $(d-1)$-dimensions. Since local bulk operators can be expressed as boundary operators smeared over either the entire spatial slice or compact spatial subregions \cite{Hamilton:2006az,Hamilton2005}, a single bulk operator can be reconstructed in different subregions \cite{Almheiri:2014lwa}. Quantum error correction provides a convenient setup where bulk operators are defined only on a code subspace of the physical Hilbert space of the conformal field theory.
In order to resolve apparent inconsistencies with space-like commutativity of local operators in quantum field theory, bulk reconstruction was studied in the context of quantum error correcting codes \cite{Almheiri:2014lwa}. Using the Ryu--Takayangi formula, \cite{Jafferis:2015del} showed that the relative entropy of nearby states computed in a boundary subregion is equivalent to the relative entropy computed in the dual entanglement wedge \cite{albion}, up to corrections on the order of Newton's constant $G_N$. These results were used in \cite{Harlow:2018fse,DongHarlowWall} to argue that CFT operators in a boundary subregion can be used to reconstruct bulk operators in the entanglement wedge.

Much of the literature on entanglement entropy contains assumptions that are only true for quantum mechanical systems with finite-dimensional Hilbert spaces. For instance, entanglement entropy has  often been defined by assuming that the Hilbert space $\calh$ can be written as $\calh = \calh_A \otimes \calh_{A^c}$, where $A$ refers to a subregion of space and $A^c$ refers to the complement of $A$. The entanglement entropy is the von Neumann entropy of the reduced density matrix one obtains after performing a partial trace on the Hilbert space $\calh_{A^c}$. However, the infinite-dimensional Hilbert space $\calh$ does not factorize in this way because the entanglement entropy contains a universal area-law divergence \cite{Witten:2018zxz}.

Von Neumann algebras are a mathematical structure that arise naturally in quantum field theory. Instead of assuming that the Hilbert space factorizes, we should characterize a causally complete region of spacetime\footnote{The causal complement of a region $R$, denoted by $R^\prime$, is defined to be all of the points in spacetime which are spacelike separated from every point in $R$. A region $R$ is causally complete if $R^{\prime \prime} = R$. Note that any von Neumann algebra $M$ satisfies $M = M^{\prime \prime}$, where the $\prime$ denotes the commutant.} by an associated von Neumann algebra \cite{Haag}. Formulating quantum field theory with von Neumann algebras is powerful because it allows one to make use of the mathematical machinery of Tomita-Takesaki theory to study entanglement. The modular operator is an important object in Tomita-Takesaki theory, and Araki \cite{Araki} has used it to define relative entropy in quantum field theory.  Theorem \ref{thm:modflow}, a central result of Tomita-Takesaki theory, formalizes the notion of modular flow. A demonstration of how von Neumann algebras are associated with the left and right Rindler wedges of Minkowski space was provided by Bisognano and Wichmann in \cite{BisoWich}. More recently, an explicit computation of mutual information for free fermions in 1+1 dimensions was performed in \cite{LongoXu}.

Given the role that entanglement entropy plays in our understanding of holography and the role that von Neumann algebras play in our understanding of entanglement entropy, it is natural to ask whether statements in the bulk reconstruction literature can be recast in a way that dispenses with the fiction that the boundary Hilbert space can be written as $\calh = \calh_A \otimes \calh_{A^c}$ for an arbitrary subregion $A$. In the context of quantum error correction with finite dimensional Hilbert spaces, \cite{Harlow:2016fse} formulates and completes the equivalence of the Ryu--Takayangi formula, entanglement wedge reconstruction, and the equality of bulk and boundary relative entropies. With the exception of the Ryu--Takayangi formula, there are natural ways to generalize these statements to the case where the Hilbert space is infinite-dimensional. The Ryu--Takayangi formula, on the other hand, computes the entanglement entropy of an arbitrary subregion in the boundary field theory, which is infinite.

In this paper, we prove that in the context of quantum error correction with infinite-dimensional Hilbert spaces, the equivalence of bulk and boundary relative entropies is a necessary and sufficient condition for entanglement wedge reconstruction. This is presented more precisely in Theorem \ref{thm:maintheorem}. We define cyclic and separating states in Definitions \ref{def:cyc} and \ref{def:sep}, and relative entropy in Definition \ref{def:relent}.

\begin{thm}
\label{thm:maintheorem}
Let $u : \calh_{code}\rightarrow \calh_{phys}$ be an isometry\footnote{This means that $u$ is a norm-preserving map. The map $u$ need not be a bijection. In general, $u^\dagger u$ is the identity on $\calh_{code}$ and $uu^\dagger$ is a projection on $\calh_{phys}$. } between two Hilbert spaces. Let $M_{code}$ and $M_{phys}$ be von Neumann algebras on $\calh_{code}$ and $\calh_{phys}$ respectively. Let $M^\prime_{code}$ and $M^\prime_{phys}$ respectively be the commutants of $M_{code}$ and $M_{phys}$. Suppose that the set of cyclic and separating vectors with respect to $M_{code}$ is dense in $\calh_{code}$. Also suppose that if $\ket{\Psi} \in \calh_{code}$ is cyclic and separating with respect to $M_{code}$, then $u \ket{\Psi}$ is cyclic and separating with respect to $M_{phys}$. Then the following two statements are equivalent:
\begin{description}
\item[ 1. Bulk reconstruction]
\begin{align} \nonumber
\begin{split}
\forall \calo \in M_{code}\ \forall \calo^\prime \in M_{code}^\prime, \quad 
\exists\tilde{\calo} \in M_{phys}\ \exists \tilde{\calo}^\prime \in M_{phys}^\prime\quad \text{such that}\quad\\
\forall \ket{\Theta} \in \calh_{code} \quad 
\begin{cases}
u \calo \ket{\Theta} =  \tilde{\calo} u \ket{\Theta}, \quad
&u \calo^\prime \ket{\Theta} =  \tilde{\calo}^\prime u \ket{\Theta}, \\
u \calo^\dagger \ket{\Theta} =  \tilde{\calo}^\dagger u \ket{\Theta}, \quad
&u \calo^{\prime \dagger} \ket{\Theta} = \tilde{\calo}^{\prime\dagger} u\ket{\Theta}.
\end{cases}\quad
\end{split}
\end{align}

\item[ 2. Boundary relative entropy equals bulk relative entropy]
\begin{align}\nonumber
\begin{split}
\text{For any $\ket{\Psi}$, $\ket{\Phi} \in \calh_{code}$ with $\ket{\Psi}$ cyclic }&\text{ and separating with respect to $M_{code}$,}\quad\quad\quad\\
\cals_{\Psi|\Phi}(M_{code})=\cals_{u\Psi|u\Phi}(M_{phys}),\ & \text{and} \  \cals_{\Psi|\Phi}(M_{code}^\prime)= \cals_{u\Psi|u\Phi}(M_{phys}^\prime),\\
\text{where $\cals_{\Psi|\Phi}(M)$ is the relative entropy.}\quad&
\end{split}
\end{align}
\end{description}
\end{thm}

Theorem \ref{thm:maintheorem} has two separate statements regarding bulk reconstruction and relative entropy. Early attempts to express bulk operators as nonlocal operators on the boundary were made in \cite{Hamilton:2006az,Hamilton2005}, and \cite{Almheiri:2014lwa} made the connection between bulk reconstruction and quantum error correction. The statement that relative entropy equals bulk relative entropy is due to \cite{Jafferis:2015del}.

Given the assumptions of Theorem \ref{thm:maintheorem}, $M_{code}$ may be viewed as a von Neumann subalgebra of $M_{phys}$. For a specific setting when the relative entropy of two states defined with respect to $M_{code}$ is identical to the relative entropy defined with respect to $M_{phys}$, $M_{code}$ is called a weakly sufficient subalgebra with respect to the two states. This particular case is studied in \cite{Petz}. However, Theorem \ref{thm:maintheorem} is concerned with the case when the relative entropies agree for all states in the code subspace.

For a generic local quantum field theory, the von Neumann algebra associated with any causally complete subregion is generically a type III$_1$ factor.\footnote{In Section 2 of \cite{TensorNetwork}, we justify this statement on physical grounds and review the classification of factors.} Assuming that this property of generic local QFTs applies in the bulk theory, one of the assumptions of Theorem \ref{thm:maintheorem} is no longer needed as seen in Remark \ref{rem:III1Dense} (see Section \ref{sec:connesstormer} for further discussion). % Refer to the discussion after writing it up in a later section.
%{\color{red} We still haven't addressed the question: why are the local algebras of QFT \emph{factors} ? Answer: follows from Haag duality.}

\begin{rem}
If $M_{code}$ and $M_{code}^\prime$ are both type III$_1$ factors, then a result of Connes--St{\o}rmer \cite{ConnesStormer} allows us to relax the assumption in Theorem \ref{thm:maintheorem} that the set of cyclic and separating vectors with respect to $M_{code}$ is dense in $\calh_{code}$.
\label{rem:III1Dense}
\end{rem}

The Reeh--Schleider theorem implies that in quantum field theory, cyclic and separating states with respect to a local algebra are dense in the Hilbert space. Likewise, if the local algebras are type III$_1$ factors, the result of Connes--St{\o}rmer also implies that cyclic and separating states are dense. This result strengthens the relevance of type III$_1$ factors to generic local quantum field theories. 

The proof of Theorem \ref{thm:maintheorem} requires two parts: statement 1 implies statement 2, and statement 2 implies statement 1 as well. Unlike the other direction, our proof that statement 1 implies statement 2 does not requite all of the assumptions of the theorem. We highlight this by presenting Theorem \ref{thm:otherdirection}:

\begin{thm}
\label{thm:otherdirection}
Let $u : \calh_{code}\rightarrow \calh_{phys}$ be an isometry between two Hilbert spaces. Let $M_{code}$ and $M_{phys}$ be von Neumann algebras on $\calh_{code}$ and $\calh_{phys}$ respectively. Let $M^\prime_{code}$ and $M^\prime_{phys}$ respectively be the commutants of $M_{code}$ and $M_{phys}$. 

\noindent Suppose that

\begin{itemize}
\item There exists some state $\ket{\Omega} \in \calh_{code}$ such that $u\ket{\Omega} \in \calh_{phys}$ is cyclic and separating with respect to $M_{phys}$. 
\item $\forall \calo \in M_{code}\ \forall \calo^\prime \in M_{code}^\prime, \quad 
\exists\tilde{\calo} \in M_{phys}\ \exists \tilde{\calo}^\prime \in M_{phys}^\prime$ such that
	\begin{align} \nonumber
	\begin{split}
	\forall \ket{\Theta} \in \calh_{code} \quad 
	\begin{cases}
	u \calo \ket{\Theta} =  \tilde{\calo} u \ket{\Theta}, \quad
	&u \calo^\prime \ket{\Theta} =  \tilde{\calo}^\prime u \ket{\Theta}, \\
	u \calo^\dagger \ket{\Theta} =  \tilde{\calo}^\dagger u \ket{\Theta}, \quad
	&u \calo^{\prime \dagger} \ket{\Theta} = \tilde{\calo}^{\prime\dagger} u \ket{\Theta}.
	\end{cases}
	\end{split}
	\end{align}
\end{itemize}

\noindent Then, for any $\ket{\Psi}$, $\ket{\Phi} \in \calh_{code}$ with $\ket{\Psi}$ cyclic and separating with respect to $M_{code}$, 
\begin{itemize}
\item $u \ket{\Psi}$ is cyclic and separating with respect to $M_{phys}$ and $M_{phys}^\prime$,
\item $\cals_{\Psi|\Phi}(M_{code})= \cals_{u\Psi|u\Phi}(M_{phys}), \quad \cals_{\Psi|\Phi}(M_{code}^\prime)= \cals_{u\Psi|u\Phi}(M_{phys}^\prime),$
\end{itemize}
where $\cals_{\Psi|\Phi}(M)$ is the relative entropy.
\end{thm}

Theorem \ref{thm:maintheorem}, our main result, has a natural interpretation in the context of AdS/CFT. As the notation suggests, $\calh_{code}$ may be interpreted as a code subspace of the physical Hilbert space $\calh_{phys}$ that consists of states with semi-classical bulk duals. The von Neumann algebra $M_{phys}$ denotes an algebra of boundary operators associated with a subregion on the boundary, and $M_{code}$ denotes an algebra of bulk operators associated with the dual entanglement wedge. The commutant algebra $M_{phys}^\prime$ is associated with the complementary boundary subregion, and $M_{code}^\prime$ is associated with the complement of the entanglement wedge of $M_{code}$.    

Theorem \ref{thm:maintheorem} provides a necessary and sufficient criterion for a subalgebra of bulk operators and its commutant to respectively be reconstructed in a subregion in the boundary and its complement. We need \cite{Jafferis:2015del} to argue that $M_{code}$ and $M_{code}^\prime$ are associated with entanglement wedges. While Theorem \ref{thm:maintheorem} may not come as a surprise to readers familiar with \cite{Harlow:2016fse,DongHarlowWall}, we emphasize that studying the infinite-dimensional case can potentially yield new physical insights in AdS/CFT. As an example in quantum field theory, the Reeh--Schlieder Theorem \cite{RS} cannot be anticipated by studying a finite-dimensional spin-lattice model where the Hilbert space factorizes as $\calh = \calh_1 \otimes \calh_2 \otimes \cdots \otimes \calh_{N}$ where $\calh_i$ denotes the finite-dimensional Hilbert space at each site.

While proving Theorem \ref{thm:otherdirection}, we show in equation \eqref{eq:modularoperators} that the modular operators associated with the bulk and boundary subregions act the same way on $\calh_{code}$. Furthermore, while proving bulk reconstruction in Theorem \ref{thm:maintheorem}, we explicitly show how to define a boundary operator that represents a given bulk operator on the code subspace. In Section \ref{sec:discussion}, we discuss the implications of the Reeh--Schlieder Theorem for infinite- and finite-dimensional quantum error correction and make contact with the results of \cite{Harlow:2016fse}. 

An outline of our proof of Theorem \ref{thm:otherdirection} is the following. 
\begin{itemize}
	\item We prove that for any $\ket{\Psi} \in \calh_{code}$ which is cyclic and separating with respect to $M_{code}$, $u \ket{\Psi}$ is cyclic and separating with respect to $M_{phys}$.\footnote{This is because we may act with an operator in $M_{code}$ to send $\ket{\Psi}$ to a vector arbitrarily close to $\ket{\Omega}$, and we may act with an operator in $M_{phys}$ to send $u\ket{\Omega}$ arbitrarily close to any vector in $\calh_{phys}$.} If such is false, the relative entropy between $u\ket{\Psi}$ and $u\ket{\Phi}$ would not be possible to be defined, as the relative modular operator requires that $u\ket{\Psi}$ be cyclic and separating with respect to $M_{phys}$.
	\item Using the fact that $M_{phys}$ and $M_{phys}^\prime$ are commutants of each other	, we show that for any $\calp \in M_{phys}$, $u^\dagger \calp u \in M_{code}$.
	\item Let $S^c_{\Psi | \Phi}$ and $S^p_{u \Psi | u \Phi}$ denote 	relative Tomita operators defined with respect to $M_{code}$ and $M_{phys}$ respectively. We relate   $S^c_{\Psi | \Phi}$ and $S^p_{u \Psi | u \Phi}$ and derive $u S^c_{\Psi | \Phi} = S^p_{u \Psi | u \Phi} u$ for generically unbounded operators. In particular, we show that their domains are equal and $S^p_{u \Psi | u \Phi}$ restricted to the vector space $(\text{Im } u)^\perp$ has a range contained within $(\text{Im } u)^\perp$.
	\item We derive a relation for the relative modular operators  associated with $S^c_{\Psi | \Phi}$ and $S^p_{u \Psi | u \Phi}$.\footnote{With the relation for the Tomita operators we derived above, we obtain a relation for the relative modular operators $\Delta^c_{\Psi | \Phi}$ and $\Delta^p_{u \Psi | u \Phi}$ to be $u\Delta^c_{\Psi | \Phi} =  \Delta^p_{u \Psi | u \Phi} u$.} This is related to the physical notion that bulk modular flow is dual to boundary modular flow. Likewise, we show that $\Delta^p_{u \Psi | u \Phi}$ restricted to the vector space $(\text{Im } u)^{\perp}$ has a range contained within $(\text{Im } u)^{\perp}$.
	\item Using the spectral theorem, we show that the spectral projections commute with the projector $uu^\dagger$.\footnote{We apply the spectral theorem separately for the restriction of $\Delta^p_{u \Psi | u \Phi}$ to $\text{Im } u$ and $(\text{Im } u)^\perp$.} We derive that the spectral projections of $\Delta^c_{\Psi | \Phi}$ are given by $u^\dagger P^p_\Omega u$, where $P^p_\Omega$ denotes the spectral projections of $\Delta^p_{u \Psi | u \Phi}$.\footnote{We use the relation $\Delta^c_{\Psi | \Phi} = u^\dagger \Delta^p_{u \Psi | u \Phi} u$. For the projections, $\Omega$ denotes a measurable subset of $\mathbb{R}$}
	\item Any function of $\Delta^p_{u \Psi | u \Phi}$ or $\Delta^c_{\Psi | \Phi}$ can be constructed once the spectral projections are known. It follows that $\braket{\Psi|\log \Delta^c_{\Psi | \Phi}|\Psi} = \braket{u\Psi|\log \Delta^p_{u\Psi | u\Phi} |u \Psi}$, and thus the relative entropies are equal. 	
\end{itemize}

We note that Theorem \ref{thm:otherdirection} dictates that statement 1 of Theorem \ref{thm:maintheorem} implies statement 2  of Theorem \ref{thm:maintheorem}. A sketch of our proof of the converse is the following. This completes the proof of Theorem \ref{thm:maintheorem}.
\begin{itemize}
	\item For any $\ket{\Phi} \in \calh_{code}$ that is cyclic and separating with respect to $M_{code}$, and for any unitary $U^\prime \in M_{code}^\prime$, the properties of relative entropy and the assumptions of the theorem imply that $0 = \cals_{\Phi | U^\prime \Phi}(M_{code}) = \cals_{u\Phi | uU^\prime \Phi}(M_{phys})$.
	\item Following the logic of \cite{Witten:2018zxz}, one may show that $\braket{ u U^\prime \Phi |  \calp u U^\prime \Phi } = \braket{u \Phi |\calp | u \Phi}$ for all $\calp \in M_{phys}$. Using the assumption that cyclic and separating states with respect to $M_{code}$ are dense in $\calh_{code}$, it follows that $u^\dagger \calp u U^\prime = U^\prime u^\dagger \calp u$. The same logic also implies that for $\calp^\prime \in M_{phys}^\prime$ and any unitary $U \in M_{code}$, $u^\dagger \calp^\prime u U = U u^\dagger \calp^\prime u$.
	\item We define a linear map $X^{\prime \, \Phi \, U^\prime} : \calh_{phys}\rightarrow \calh_{phys}$ by $X^{\prime \, \Phi \, U^\prime}  \calp u \ket{\Phi} := \calp u U^\prime \ket{\Phi} \quad \forall \calp \in M_{phys}$, and we show that $X^{\prime \, \Phi \, U^\prime}$ is unitary and that $X^{\prime \, \Phi \, U^\prime} \in M_{phys}^\prime$.
	\item Since $u^\dagger X^{\prime \, \Phi \, U^\prime} u U = U u^\dagger X^{\prime \, \Phi \, U^\prime} u$ and any operator in $M_{code}$ may be written as a linear combination of four unitary operators in $M_{code}$, we show that $u^\dagger X^{\prime \, \Phi \, U^\prime} u = U^\prime$. We also show that $X^{\prime \, \Phi \, U^\prime}$ maps the vector space $\text{Im } u \rightarrow \text{Im } u$. Hence, $ X^{\prime \, \Phi \, U^\prime} u = u U^\prime$ 
	\item Using similar methods, we then show that $ (X^{\prime \, \Phi \, U^\prime})^{\dagger} u = u (U^\prime)^\dagger$. Thus, the unitary operator $U^\prime \in M_{code}$ may be reconstructed as $X^{\prime \, \Phi \, U^\prime}$ for some choice of $\ket{\Phi} \in \calh_{code}$ that is cyclic and separating with respect to $M_{code}$.  
	\item Since any operator in $M_{code}$ may be written as a linear combination of four unitary operators in $M_{code}$, we have a way to represent any operator in $M_{code}$ as an operator in $M_{phys}$. The same logic applies to show that any operator in $M_{code}^\prime$ may be represented as an operator in $M_{phys}^\prime$.
\end{itemize}

The rest of this paper is summarized as follows. In Section 2, we define von Neumann algebras and functions of operators, and we review the spectral theorem (for unbounded operators). In Section 3, we review the relative modular operator from Tomita-Takesaki theory, and define the relative entropy. In Section 4, we prove that when the bulk reconstruction is satisfied, the relative entropy is equivalent between the boundary and the bulk (Theorem \ref{thm:otherdirection}). In Section \ref{sec:converse}, we prove the converse, completing the proof of Theorem \ref{thm:maintheorem}. In Section \ref{sec:discussion}, we physically interpret Theorem \ref{thm:maintheorem} and relate our work to previous work on finite-dimensional quantum error correction and holography.

\section{Bounded and Unbounded Operators}
In this section, we review some results in functional analysis that are used in the proofs of Theorems \ref{thm:otherdirection} and \ref{thm:maintheorem}. In particular, we explain how to define a function of a bounded self-adjoint operator and we review the spectral theorem (for unbounded operators).  We mostly follow reference \cite{ReedSimon}.

\begin{defn}
	An \emph{operator} on a Hilbert space $\calh$ is a linear map from its domain, a linear subspace of $\calh$, into $\calh$.
\end{defn}
\begin{defn}
	A \emph{bounded operator} is an operator $\calo$ that satisfies $||\calo \ket{\psi}|| \leq K ||\ket{\psi}|| \quad \forall \ket{\psi} \in \calh$ for some $K \in \mathbb{R}$. We let $\calb(\calh)$ denote the algebra of bounded operators on $\calh$.
\end{defn}

\begin{defn}
	The {\em commutant} of a subset $S\subset \mathcal{B}(\mathcal{H})$ is the set $S^\prime$ of bounded operators that commute with all operators in $S$, i.e. $S^\prime =\{\calo \in \mathcal{B}(\mathcal{H}) : \calo \calp = \calp \calo \ \forall \calp \in S \}$. The double commutant of $S$ is the commutant of $S^\prime$.
\end{defn}
\begin{defn}
	A \emph{von Neumann algebra} is an algebra of bounded operators that contains the identity operator, is closed under hermitian conjugation, and is equal to its double commutant.  
\end{defn}

\begin{thm}
	\label{thm:continuity}
	Let $\calo \in \calb(\calh)$. Let $\{\ket{\Psi_n}\} \in \calh$ be a sequence of vectors such that its limit vanishes, i.e. $\lim_{n \rightarrow \infty} \ket{\Psi_n} = 0$. Then, $\lim_{n \rightarrow \infty} \calo \ket{\Psi_n} = 0$.
\end{thm}
Theorem \ref{thm:continuity} implies that bounded operators define a continuous linear map on the Hilbert space. Any bounded operator that annihilates a dense subspace of the Hilbert space is identically zero.
\begin{defn}
	A \emph{densely defined operator} on a Hilbert space $\calh$ is an operator whose domain is a dense subspace of $\calh$. 
\end{defn}

\subsection{Functions of bounded operators}

In this section, we will explain how to understand functions of bounded operators.

\begin{defn}
	The \emph{spectrum} of $\calo \in \calb(\calh)$ is defined as
	\begin{equation*}
	\sigma(\calo) := \{ \lambda \in \mathbb{C} : \calo - \lambda I \text{ is not invertible}\},
	\end{equation*}
	where $I$ denotes the identity operator.
\end{defn}
We will make use of the mathematical facts that $\sigma(\calo)$ is a nonempty closed bounded subset of $\mathbb{C}$ and that when $\calo$ is self-adjoint, $\sigma(\calo) \subset \mathbb{R}$ and
$ ||\calo|| = \sup_{\lambda \in \sigma(\calo)}|\lambda| $ \cite{Jones-vNalg}\cite{ReedSimon}.

\begin{defn}
	Let $\calo \in \calb(\calh)$ be a self-adjoint operator. We denote the set of continuous $\mathbb{R}$-valued functions defined on $\sigma(\calo)$ by $C(\sigma(\calo))$.
\end{defn}

\begin{defn}
	For every self-adjoint operator $\calo \in \calb(\calh)$, we define the $L_{\infty}$ norm (denoted by $|| \cdot ||_{\infty}$) of $f \in C(\sigma(\calo))$ by
	\begin{equation*}
	||f||_{\infty} = \sup_{x \in \sigma(\calo)}|f(x)|.
	\end{equation*} 
\end{defn}

\begin{thm}[\cite{ReedSimon}, page 121]
	Given a self-adjoint operator $\calo \in \calb(\calh)$, the set of polynomials (with $\mathbb{R}$-valued coefficients) is dense in $C(\sigma(\calo))$ in the $L_{\infty}$ norm.
\end{thm}

\begin{defn}
	For any polynomial $p(x) = \sum_{n = 0}^N a_n x^n$ with $a_n \in \mathbb{R}$, we define $p(\calo) := \sum_{n = 0}^N a_n \calo^n$ for $\calo \in \calb(\calh)$. 
\end{defn}

\begin{thm}[\cite{ReedSimon}, page 223]
	Let $p(x) = \sum_{n = 0}^N a_n x^n$ with $a_n \in \mathbb{R}$. Let $\calo \in \calb(\calh)$.\footnote{Note that $\calo$ need not be self-adjoint.} Then
	\[ \sigma(p(\calo)) = \{p(\lambda)|\lambda \in \sigma(\calo) \}. \]
\end{thm}

\begin{thm}[\cite{ReedSimon}, page 223]
	For any self-adjoint operator $\calo \in \calb(\calh)$ and any polynomial $p \in C(\sigma(\calo))$,
	\begin{equation*}
	||p(\calo)|| = ||p||_{\infty}. 
	\end{equation*}
\end{thm}
\begin{proof}
	$||p(\calo)|| = \sup_{\lambda \in \sigma(p(\calo))} |\lambda| = \sup_{\lambda \in \sigma(\calo)} |p(\lambda)|= ||p||_{\infty}$.
\end{proof}

Let $\calo \in \calb(\calh)$ be self-adjoint. Let $P$ denote the space of polynomials defined on $\mathbb{R}$ with $\mathbb{R}$-valued coefficients. Define a map $\tilde{\phi}_\calo : P \rightarrow \calb(\calh)$ such that $\tilde{\phi}_\calo(p) = p(\calo)$ for any polynomial $p \in P$. The map $\tilde{\phi}_\calo$ is a bounded linear operator because $||\tilde{\phi}_\calo(p)|| = ||p||_\infty$. Hence, $\tilde{\phi}_\calo$ may be uniquely extended to a bounded linear operator $\phi_\calo : C(\sigma(\calo)) \rightarrow \calb(\calh)$. For $f \in C(\sigma(\calo))$, we define $f(\calo) := \phi_\calo(f)$. If $\{p_n\} \in P$ denotes a sequence of polynomials such that $\lim_{n \rightarrow \infty} p_n = f$ (where the limit converges in the $L_{\infty}$ norm), then we may also write
\begin{equation}
f(\calo) = \lim_{n \rightarrow \infty} p_n(\calo),
\end{equation}
where the limit converges in the norm topology. If $f,g \in C(\sigma(\calo))$, then one may show \cite{ReedSimon} that $f(\calo)g(\calo) = (fg)(\calo)$ and that $(f^*)(\calo) = f(\calo)^\dagger$.

\begin{thm}[\cite{Jones-vNalg}, page 19]
	\label{thm:fourunitaries}
	Let $M$ be a von Neumann algebra. Any operator in $M$ is a linear combination of four unitary operators in $M$.
\end{thm}
\begin{proof}
	Let $\calo \in M$. We may write\[ \calo = \frac{1}{2}(\calo + \calo^\dagger) -\frac{i}{2} (i(\calo - \calo^\dagger)). \]
	This shows that $\calo$ may be written as a linear combination of two self-adjoint operators in $M$. Next, let $\mathcal{Q} \in M$ be a self-adjoint operator that satisfies $||\mathcal{Q}|| < 1$. The condition $||\mathcal{Q}|| < 1$ is important because the function $f(x) = \sqrt{1 - x^2}$ is $\mathbb{R}$-valued and continuous only for $|x| < 1$. Define $U := \mathcal{Q} + i \sqrt{1 - \mathcal{Q}^2}$. Then $U$ is unitary, $U \in M$, and $\mathcal{Q} =  \frac{U + U^\dagger}{2}$.
 \end{proof}

\subsection{Unbounded operators}
Unbounded operators are generically not defined on the entire Hilbert space. The domain of an operator $\calo$ is denoted by $D(\calo)$. The definition of $\calo^\dagger$ is subtle when $\calo$ is unbounded, as $\calo^\dagger$ may not be defined on the entire Hilbert space.
\begin{defn}
	A densely defined operator $\calo$ is \emph{closed} when $\calo (\lim_{n \rightarrow \infty} \ket{\psi_n}) = \lim_{n \rightarrow \infty} \calo \ket{\psi_n}$ whenever both limits exist.
\end{defn}
\begin{defn}
	Let $\calo$ be a densely defined operator on $\calh$. The domain of the adjoint $\calo^\dagger$ is defined by
	\[ D(\calo^\dagger) = \{\ket{\phi} : \exists \ket{\eta} \in \calh \text{ such that } \braket{\phi|\calo|\psi}= \braket{\eta|\psi} \quad \forall \ket{\psi} \in D(\calo) \}. \]
	For $\ket{\phi} \in D(\calo^\dagger)$ there is precisely one $\ket{\eta}$ that satisfies the above criteron. We define
	\[ \calo^\dagger \ket{\phi} := \ket{\eta}. \]
\end{defn}
\begin{thm}[\cite{ReedSimon}, page 252]
	If $\calo$ is a densely defined operator on $\calh$, then $\calo^\dagger$ is closed. If $\calo$ is closed, $D(\calo^\dagger)$ is dense in $\calh$.
\end{thm}
\begin{defn}
	A densely defined operator $\calo$ is \emph{self-adjoint} when $\calo = \calo^\dagger$. In particular, $D(\calo) = D(\calo^\dagger)$.
\end{defn}
\begin{defn}
	A densely defined operator is \emph{positive}  when $\braket{\psi|\calo|\psi} \ge 0 \quad \forall \ket{\psi} \in D(\calo)$.
\end{defn}

\begin{defn}
	Let $\calo$ be a closed operator on a Hilbert space $\calh$. $\lambda \in \mathbb{C}$ is in the \emph{resolvent set} of $\calo$ if $\lambda I - \calo$ is a bijection of $D(\calo)$ onto $\calh$. The \emph{spectrum} of $\calo$, denoted $\sigma(\calo)$, is defined to be the set of all complex numbers that are not in the resolvent set of $\calo$.
\end{defn}
\begin{thm}
	Let $\calo$ be a self-adjoint positive operator. Then the spectrum of $\calo$ is a subset of $[0,\infty)$.   
\end{thm}
\begin{proof}
For any $\ket{\chi} \in D(\calo)$ and any $\lambda = \lambda_1 + i \lambda_2$ for $\lambda_1,\lambda_2 \in \mathbb{R}$, note that\footnote{To be explicit, we have that
\[ \braket{(\calo - \lambda I)\chi|(\calo - \lambda I)\chi} = \braket{(\calo - \lambda_1 I)\chi|(\calo - \lambda_1 I)\chi} + \braket{\lambda _2\chi| \lambda_2 \chi} + i \lambda_2 \braket{\chi|(\calo - \lambda_1 I )\chi} - i \lambda_2 \braket{ (\calo - \lambda_1I)\chi |\chi}.\] The last two terms cancel because $\calo$ is self-adjoint and $\lambda_1$ is real.}
\begin{equation}
||(\calo - \lambda I)\ket{\chi}||^2 = \lambda_2^2 ||\ket{\chi}||^2 + ||(\calo - \lambda_1) \ket{\chi}||^2 \ge \lambda_2^2 ||\ket{\chi}||^2.
\end{equation}
Let us consider the case when $\lambda_2 \neq 0$. Then $\ker(\calo - \lambda I) = \{0\}$ so that $\calo - \lambda I$ is an injection. Using the fact that $D(\calo)$ is dense in $\calh$, one may show that the orthocomplement of the range of $(\calo - \lambda I)$ is trivial, implying that the range of $(\calo - \lambda I)$ is dense in $\calh$. Then, the previous equation implies that if $\{\ket{\chi_n} \} \in D(\calo)$ is a sequence such that $\lim_{n \rightarrow \infty} (\calo - \lambda I)\ket{\chi_n}$ exists, then $\lim_{n \rightarrow \infty} \ket{\chi_n}$ also exists. Since $\calo$ is a closed operator, the range of $(\calo - \lambda I)$ is also closed. Thus, $(\calo - \lambda I)$ is a bijection from $D(\calo)$ onto $\calh$, demonstrating that $\lambda$ is in the resolvent set of $\calo$.

Now, consider the case when $\lambda \in \mathbb{R}$, $\lambda < 0$. For any $\ket{\chi} \in D(\calo)$,
\begin{equation}
||(\calo - \lambda I)\ket{\chi}||^2 = |\lambda|^2 ||\ket{\chi}||^2 - 2 \braket{\chi|\calo|\chi}  \lambda + ||\calo\ket{\chi}||^2.
\end{equation}
As $\calo$ is a positive operator,
\begin{equation}
||(\calo - \lambda I)\ket{\chi}||^2 \geq |\lambda|^2 ||\ket{\chi}||^2 .
\end{equation}
The same logic used in the previous case establishes that $\lambda$ is in the resolvent set of $\calo$. Hence, the spectrum of $\calo$ must be a subset of $[0,\infty)$. 
\end{proof}

\begin{thm}[\cite{ReedSimon}, page 316]
Let $\calo$ be a closed operator. Then $D(\calo^\dagger \calo) = \{\ket{\psi} : \ket{\psi} \in D(\calo), \ \calo\ket{\psi} \in D(\calo^\dagger) \}$ is dense in the Hilbert space and $\calo^\dagger \calo$ is self-adjoint and positive.
\end{thm}

\subsection{The spectral theorem for unbounded operators}

In this section, we closely follow \cite{ReedSimon} (pages 262-263), to which we refer the reader for more details on the spectral theorem. Note that a projection $P \in \calb(\calh)$ is idempotent and hermitian i.e. $P = P^2 = P^\dagger$. 

\begin{defn}
	A \emph{projection valued measure} assigns a projection $P_\Omega$ to every Borel set $\Omega \subset \mathbb{R}$ such that 
\begin{itemize}
	\item $P_{\emptyset} = 0$, $P_{(-\infty,\infty)} = I$
	\item $P_{\Omega_1} P_{\Omega_2} = P_{\Omega_1 \cap \Omega_2}$
	\item If $\Omega = \cup_{n = 1}^\infty \Omega_n$ with $\Omega_n \cap \Omega_{m} = \emptyset$ if $n \neq m$, then $P_\Omega$ is a strong limit of $\sum_{n = 1}^N P_{\Omega_n}$. 
\end{itemize}
\end{defn}

Given any vector $\ket{\psi} \in \calh$, $\braket{\psi|P_\Omega|\psi}$ defines an integration measure for Borel functions, which we will use in Definition \ref{def:log}. 

\begin{thm}[{\bf Spectral Theorem} \cite{ReedSimon}, page 263]
	There is a one-to-one correspondence between self-adjoint operators $\calo$ and projection valued measures $P^\calo_\Omega$. The correspondence is given by
	\[ \calo = \int_\mathbb{R} \lambda \, d(P^\calo_\lambda). \]
	The notation means that we are integrating the function $f(\lambda) = \lambda$ on $\mathbb{R}$ with the projection-valued measure given by $P^\calo_\Omega$. The integral converges strongly.\footnote{For any $\ket{\psi} \in D(\calo)$, the integral $\int_\mathbb{R} \lambda d(P^\calo_\lambda \ket{\psi})$ with vector-valued measure $P^\calo_\Omega \ket{\psi}$ converges in the Hilbert space norm to $\calo\ket{\psi}$. The integral does not converge for $\ket{\psi} \notin D(\calo)$.}
\end{thm}
Intuitively, $P_\Omega^\calo$ is the projection onto the ``eigenspace'' spanned by all ``eigenvalues'' in $\Omega$. We will need that $P^\calo_{(-\infty,\infty)} = P^\calo_{\sigma(\calo)}$, where $\sigma(\calo)$ denotes the spectrum of $\calo$. For the details on how the spectral projections associated with a self-adjoint operator $\calo$ are constructed, see Theorem VIII.4 and discussions afterwards in Section VIII.3 of \cite{ReedSimon}. 

\begin{defn}
	\label{def:log}
	Given a self-adjoint positive operator $\calo$,  the diagonal matrix element of $\log \calo$ is given by
	\[ \braket{\psi| \log \calo|\psi} := \int_{0}^\infty \log \lambda \, d (\braket{\psi|P^\calo_\lambda|\psi}), \]
	for all $\ket{\psi} \in \calh$ such that the above integral converges, where $P^\calo_\Omega$ is the unique projection valued measure associated with $\calo$ by the spectral theorem.
\end{defn}
Note that $\log x$ is continuous for $x \in (0,\infty)$. Thus, $\log x$ is a Borel function. One can define a self-adjoint operator using any real-valued Borel function on $\mathbb{R}$. See page 264 of \cite{ReedSimon}.
\begin{thm}
	\label{thm:saturation}
	Let $\calo$ be a self-adjoint positive operator. For all $\ket{\psi} \in D(\calo)$ such that $\braket{\psi|\log \calo|\psi}$ is defined,
	\[ \braket{\psi|\log \calo|\psi} \leq \braket{\psi|\calo|\psi} - \braket{\psi|\psi}, \]
	and the inequality is saturated if and only if $\calo \ket{\psi} = \ket{\psi}$. 
\end{thm}
\begin{proof}
Assume $\ket{\psi} \neq 0$.	For $x > 0$, note that $\log x \leq x - 1$. It follows that
	\begin{equation} 
\braket{\psi| \log \calo|\psi} = \int_{0}^\infty \log \lambda \, d (\braket{\psi|P^\calo_\lambda|\psi})\leq \int_{0}^\infty \lambda  \, d (\braket{\psi|P^\calo_\lambda|\psi}) - \int_{0}^\infty 1 \, d (\braket{\psi|P^\calo_\lambda|\psi}).
\label{eq:inequality}
\end{equation}
	The first integral on the right hand side converges because $\ket{\psi} \in D(\calo)$. The second integral converges to $\braket{\psi|\psi}$ because the spectrum of $\calo$ is a subset of $[0,\infty)$, which implies that $P^\calo_{[0,\infty)} = P^\calo_{(-\infty,\infty)} = I$. Hence,
		\begin{equation} 
	\braket{\psi| \log \calo|\psi} \leq \braket{\psi|\calo|\psi} - \braket{\psi|\psi} .
	\end{equation}
As $\log x \leq x - 1$ is only saturated for $x = 1$, the inequality in equation \eqref{eq:inequality} is only saturated when the measure $\braket{\psi|P^\calo_\Omega|\psi}$ is such that $\braket{\psi|P^\calo_\Omega|\psi} = 0$ when $1 \notin \Omega$. If $1 \neq \Omega$, then $\braket{\psi|P^\calo_\Omega|\psi} = \braket{P^\calo_\Omega\psi|P^\calo_\Omega\psi}$ implies that $P^\calo_\Omega \ket{\psi} = 0$. If $1 \in \Omega$, then the fact that $\int_\mathbb{R} 1 d(\braket{\psi|P^\calo_\lambda|\psi}) = \braket{\psi|\psi}$ implies that $ \braket{P^\calo_\Omega \psi|P^\calo_\Omega \psi} = \braket{\psi|P^\calo_\Omega|\psi}  = \braket{\psi|\psi}$. For $1 \in \Omega$, note that the Cauchy-Schwartz inequality $|\braket{\psi|P^\calo_\Omega|\psi}| \leq ||\ket{\psi}|| \cdot||P^\calo_\Omega \ket{\psi}||$ is saturated, which implies that $P_\Omega^\calo\ket{\psi}$ is a multiple of $\ket{\psi}$, and this multiple must be $1$. Thus, for $1 \in \Omega$, $P^\calo_\Omega \ket{\psi} = \ket{\psi}$. This implies that
\begin{equation}
\calo\ket{\psi} = \int_{\mathbb{R}} \lambda \, d(P^\calo_\lambda \ket{\psi}) = \ket{\psi}.
\end{equation}
\end{proof}

\section{Review of Tomita-Takesaki theory}

Previous works on entanglement entropy and AdS/CFT \cite{Jafferis:2015del,Casini2008,DongHarlowWall, casinitestetorroba2016} have used the definition for the relative entropy as $S(\rho,\sigma) = \text{Tr }(\rho \log \rho - \rho \log \sigma)$. Since $S(\rho,\sigma)$ does not increase upon performing a partial trace on $\rho$ and $\sigma$, the relative entropy may be intuitively thought of as a measure of distinguishability between two states. Araki's definition of the relative entropy \cite{Araki} also has a monotonicity property, and it reduces to $S(\rho,\sigma)$ when the Hilbert space is finite-dimensional \cite{Witten:2018zxz}. Hence, we might expect that statements about relative entropy in AdS/CFT can be reformulated for infinite-dimensional Hilbert spaces. 

We want to understand the connection between entanglement wedge reconstruction and the equivalence of bulk and boundary relative entropies in infinite dimensional Hilbert spaces, using Tomita-Takesaki theory. Tomita-Takesaki theory provides us with the relative modular operator which is used to define the relative entropy. In this section, we review properties of the relative modular operator and the definition of the relative entropy, following \cite{Araki,Witten:2018zxz, Jones-vNalg}.

\begin{defn}
	\label{def:cyc}
	A vector $\ket{\Psi} \in \calh$ is said to be {\em cyclic} with respect to a von Neumann algebra $M$ when the set of vectors $\calo\ket{\Psi}$ for $\calo \in M$ is dense in $\calh$. 
\end{defn}
\begin{defn}
	\label{def:sep}
%	A vector $\ket{\Psi} \in \calh$ is said to be {\em separating} with respect to a von Neumann algebra $M$ when zero is the only operator in von Neumann algebra that annihilates the vector,  (i.e. $\calo\ket{\Psi} = 0 \implies \calo = 0$ for $\calo \in M$). 
	A vector $\ket{\Psi} \in \calh$ is {\em separating} with respect to a von Neumann algebra $M$ when zero is the only operator in $M$ that annihilates $\ket{\Psi}$. That is, $\calo\ket{\Psi} = 0 \implies \calo = 0$ for $\calo \in M$. 
\end{defn}

Given a von Neumann algebra $M \subset \calb(\calh)$ and a vector $\ket{\Psi} \in \calh$, we may define a map $e_\Psi:M \rightarrow \calh : \calo \mapsto \calo \ket{\Psi}$. $\calh$ is the closure of the image of $e_\Psi$ iff $\ket{\Psi}$ is cyclic with respect to $M$. Also, $\ker e_\Psi = \{0\}$\footnote{In other words, $e_\Psi$ is an injective map.} iff $\ket{\Psi}$ is separating with respect to $M$.

\begin{defn}
Let $\ket{\Psi},\ket{\Phi} \in \calh$ and $M$ be a von Neumann algebra. 
The {\em relative Tomita operator} is the operator $S_{\Psi | \Phi}$ that acts as
\begin{equation*}
S_{\Psi|\Phi} \ket{x} := \ket{y}
\end{equation*}
for any sequence $\{\calo_n\} \in M$ such that the limits $\ket{x} = \lim_{n \rightarrow \infty} \calo_n \ket{\Psi}$ and $\ket{y} = \lim_{n \rightarrow \infty} \calo_n^\dagger \ket{\Phi}$ both exist.
\end{defn}

The relative Tomita operator $S_{\Psi | \Phi}$ is well-defined on a dense subset of the Hilbert space if and only if $\ket{\Psi}$ is cyclic and separating with respect to $M$.\footnote{$S_{\Psi | \Phi}$ is well-defined if and only if $\lim_{n \rightarrow \infty} \calo_n \ket{\Psi} = 0 \implies \lim_{n \rightarrow \infty} \calo_n^\dagger \ket{\Psi} = 0$. See footnote 14 of \cite{Witten:2018zxz} for a proof of why this is true. $S_{\Psi | \Phi}$ is densely defined because $\ket{\Psi}$ is cyclic with respect to $M$.} Note that $S_{\Psi|\Phi}$ is a closed operator.

\begin{thm}[\cite{Jones-vNalg}, page 94]
	Let $\ket{\Psi},\ket{\Phi} \in \calh$ both be cyclic and separating with respect to a von Neumann algebra $M$. Let $S_{\Psi|\Phi}$ and $S^\prime_{\Psi|\Phi}$ be the relative Tomita operators defined with respect to $M$ and its commutant $M^\prime$ respectively. Then
	\begin{equation}
	S_{\Psi|\Phi}^\dagger = S^\prime_{\Psi|\Phi}, \ S_{\Psi|\Phi}^{\prime \, \dagger} = S_{\Psi|\Phi}.
	\end{equation}
\end{thm}

\begin{defn}
Let $S_{\Psi|\Phi}$ be a relative Tomita operator and $\ket{\Psi} \in \calh$ be cyclic and separating with respect to a von Neumann algebra $M$. The {\em relative modular operator} is $$\Delta_{\Psi|\Phi} := S_{\Psi|\Phi}^\dagger S_{\Psi|\Phi}.$$
\end{defn}

If $\ket{\Phi}$ is replaced with $\calo^\prime \ket{\Phi}$, where $\calo^\prime \in M^\prime$ is unitary, then the relative modular operator remains unchanged \cite{Witten:2018zxz}:
\begin{equation} \Delta_{\Psi|\Phi} = \Delta_{\Psi|\calo^\prime\Phi}. \end{equation}

\begin{defn}[\cite{Araki}] \label{def:relent}
Let $\ket{\Psi},\ket{\Phi} \in \calh$ and $\ket{\Psi}$ be cyclic and separating with respect to a von Neumann algebra $M$. Let $\Delta_{\Psi|\Phi}$ be a relative modular operator. The {\em relative entropy} with respect to $M$ of $\ket{\Psi}$ is
	\begin{equation*} 
	\cals_{\Psi|\Phi}(M) = - \braket{\Psi|\log \Delta_{\Psi | \Phi}|\Psi}. 
	\end{equation*} 
\end{defn}
Note that the relative entropy $\cals_{\Psi|\Phi}(M)$ is nonnegative and it vanishes precisely when $\ket{\Phi} = \calo^\prime \ket{\Psi}$ for a unitary $\calo^\prime \in M^\prime$.

\begin{defn}
	Let $M$ be a von Neumann algebra on $\calh$ and $\ket{\Psi}$ be a cyclic and separating vector for $M$. The \emph{Tomita operator} $S_\Psi$ is
	$$ S_\Psi := S_{\Psi | \Psi},$$ where $S_{\Psi | \Psi}$ is the relative modular operator defined with respect to $M$. The \emph{modular operator} $\Delta_\Psi = S_\Psi^\dagger S_\Psi$ and the \emph{antiunitary operator} $J_\Psi$ are the operators that appear in the polar decomposition of $S_\Psi$ such that
	$$S_\Psi = J_\Psi \Delta_{\Psi}^{1/2}.$$
\end{defn}

\begin{thm}[Tomita-Takesaki \cite{Takesaki}]
	\label{thm:modflow}
	Let $M$ be a von Neumann algebra on $\calh$ and let $\ket{\Psi}$ be a cyclic and separating vector for $M$. Then \begin{itemize}
		\item $J_\Psi M J_\Psi = M^\prime.$
		\item $\Delta_\Psi^{it} M \Delta_\Psi^{-it} = M \quad \forall t \in \mathbb{R}$.
	\end{itemize} 
\end{thm}

Theorem \ref{thm:modflow} is important because it allows us to interpret $\Delta_\Psi$ as the operator that generates a modular flow on $M$. Suppose that the Hilbert space $\calh$ factorizes as $\calh = \calh_\ell \otimes \calh_r$. For concreteness, we may intuitively think of $\calh_\ell$ as a Hilbert space that corresponds to the left Rindler wedge of Minkowski space, while $\calh_r$ corresponds to the right Rindler wedge. 

\begin{figure}[H]
\begin{center}
\scalebox{0.65}{
\begin{tikzpicture}
\draw[->] (-5,0)--(5,0);
\draw[->] (0,-5)--(0,5);
\draw (-4.5,-4.5)--(4.5,4.5);
\draw (-4.5,4.5)--(4.5,-4.5);
\draw[->] (5,-2) arc (250:110:2);
\draw[->] (4,-2.7) arc (240:120:3);
\draw[->] (-5,-2) arc (-70:70:2);
\draw[->] (-4,-2.7) arc (-60:60:3);
\node[draw=none,label={[label distance=0.1mm]south:$K_r$}] (R) at (5,3.5) {};
\node[draw=none,label={[label distance=0.1mm]south:$K_\ell$}] (L) at (-5,3.5) {};
\end{tikzpicture}}
\label{fig:rindlerwedges}
\caption{Two Rindler wedges in Minkowski space. The generators $K_ r$ and $K_\ell$ correspond to boosts, as shown.}
\end{center}
\end{figure}
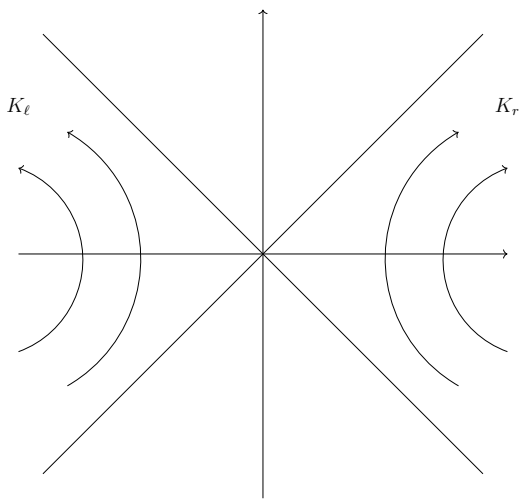

For a given state $\ket{\Psi} \in \calh$, if we define
\begin{equation}
\rho := \ket{\Psi}\bra{\Psi}, \quad \rho_\ell := \text{Tr}_r \, \rho, \quad \rho_r := \text{Tr}_\ell \, \rho,
\end{equation}
then the reduced density matrices $\rho_\ell$ and $\rho_r$ generate a modular flow on operators that act on $\calh_\ell$ and $\calh_r$, respectively. The modular operator $\Delta_\Psi$ corresponding to the von Neumann algebra that acts nontrivially on $\calh_\ell$ is then given by
\begin{equation}
\Delta_\Psi = \rho_\ell \otimes \rho_r^{-1}.
\end{equation}
When $\ket{\Psi}$ is the vacuum and $\calh_\ell$ and $\calh_r$ correspond to Rindler wedges, we have that
\begin{equation}
\rho_\ell = e^{- 2 \pi K_\ell}, \quad \rho_r = e^{- 2 \pi K_r},
\end{equation}
where $K_\ell$ and $K_r$ are the boost generators that act respectively on the left and right wedges (see Figure \ref{fig:rindlerwedges}). The modular operator $\Delta_\Psi$ is then given by
\begin{equation}
\Delta_\Psi = e^{-2\pi (K_\ell - K_r)}.
\end{equation}
In this context, Theorem \ref{thm:modflow} states that the modular flow maps operators in a Rindler wedge to operators in the same Rindler wedge. Thus, the algebraically defined modular flow in Theorem \ref{thm:modflow} has a geometric interpretation. This is an example of \emph{modular covariance}, which is the property that the modular flow is a spacetime symmetry. The unitary group generated by the modular operator associated with the vacuum state implements the Lorentz boosts that leave the Rindler wedges invariant. The antiunitary operator $J$ corresponds to the operator $CRT$, where $C$ denotes charge conjugation, $R$ is a reflection that maps one wedge into the other, and $T$ is time reversal \cite{baumgartelwollenberg}.

%algebras have to satisfy wightman axiom
%modular operator associated with vacuum state
%algebras of observables localized in wedge regions
%wightman shows it has geometric context
%ppl show that the unitary group generated by modular operator implements the lorentz boosts that leave the wedge region to be invariant

%J implements a spacetime reflection about the wedge

%\sout{In finite dimensional Hilbert spaces, that factorize as $\calh = \calh_A \otimes \calh_{A^c}$, the modular Hamiltonian is defined as $-\log \rho_A$, where $\rho_A = \text{Tr }_{A^c} \rho$. In the simple example of entanglement between the left and right Rindler wedges of Minkowski space, the modular Hamiltonian is the boost generator.} 

One of the findings of \cite{Jafferis:2015del} is that bulk modular flow is dual to boundary modular flow. As an intermediate step in proving the equivalence of bulk and boundary entropies, we will also show that the bulk and boundary modular operators act on the code subspace in the same way. This is further evidence that the definitions and theorems of Tomita-Takesaki theory are relevant for understanding bulk reconstruction.

\section{Proof of Theorem \ref{thm:otherdirection}}

This section contains the proof of Theorem \ref{thm:otherdirection}. In Lemma \ref{lem:firstlemma}, we show that cyclic and separating states in $\calh_{code}$ are mapped to cyclic and separating states in $\calh_{phys}$.
In Lemma \ref{lem:phystocode}, we relate operators in $M_{phys}$ to operators in $M_{code}$. In Section \ref{sec:specialcase}, we consider Theorem \ref{thm:otherdirection} in a special case where the relative Tomita operators are bounded. In Section \ref{sec:general}, we prove Theorem \ref{thm:otherdirection} in full generality.

\begin{lem}
\label{lem:firstlemma}
Under the assumptions of Theorem \ref{thm:otherdirection}, for every $\ket{\Psi} \in \calh_{code}$ that is cyclic and separating with respect to $M_{code}$, $u\ket{\Psi}$ is cyclic and separating with respect to $M_{phys}$.
\end{lem}
\begin{proof}
Let $\ket{\Omega}$ be defined as in Theorem \ref{thm:otherdirection}.	We will first show that $u\ket{\Psi}$ is cyclic with respect to $M_{phys}$. That is, we can act on $u \ket{\Psi}$ with an operator in $M_{phys}$ to get a state arbitrarily close to any state in $\calh_{phys}$. Given any $\ket{\Phi} \in \calh_{phys}$ and $\epsilon > 0$, we need to choose $\calp \in M_{phys}$ such that $||\ket{\Phi} - \calp u \ket{\Psi}|| < \epsilon$.
	Choose $\hat{\calp} \in M_{phys}$ such that $||\hat{\calp}u\ket{\Omega} - \ket{\Phi}|| < \frac{\epsilon}{2}$ and $\hat{\calp} \neq 0$. Since $\ket{\Psi}$ is cyclic with respect to $M_{code}$, choose $\calo \in M_{code}$ such that $||\calo \ket{\Psi} - \ket{\Omega}|| < \frac{\epsilon}{2 ||\hat{\calp}||}$. Let $\tilde{\calo} \in M_{phys}$ be an operator that satisfies $\tilde{\calo}u \ket{\Theta} = u \calo \ket{\Theta} \forall \ket{\Theta} \in \calh_{code}$. Then, note that
	\begin{equation}
\ket{\Phi} - \hat{\calp} \tilde{\calo} u \ket{\Psi}  = \ket{\Phi} - \hat{\calp} u \calo \ket{\Psi} = \ket{\Phi} - \hat{\calp} u \ket{\Omega} - \hat{\calp} u (\calo \ket{\Psi} - \ket{\Omega} ) .
\end{equation}
By the triangle inequality,
	\begin{equation}
||\ket{\Phi} - \hat{\calp} \tilde{\calo} u \ket{\Psi}||  \leq || \ket{\Phi} - \hat{\calp} u \ket{\Omega}||+|| \hat{\calp} || \cdot ||  \calo \ket{\Psi} - \ket{\Omega}  ||.
\end{equation}
By choosing $\calp = \hat{\calp} \tilde{\calo} $, we see that $u \ket{\Psi}$ is cyclic with respect to $M_{phys}$. The same logic shows that $u\ket{\Psi}$ is cyclic with respect to $M_{phys}^\prime$ and hence separating for $M_{phys}$.
\end{proof}

\begin{lem} \label{lem:phystocode}
	Under the assumptions of Theorem \ref{thm:otherdirection}, for any $\calp \in M_{phys}$, $u^\dagger \calp u \in M_{code}$.
\end{lem}
\begin{proof}
	Choose any $\calo^\prime \in M_{code}^\prime$. For any $\ket{\Psi},\ket{\Phi} \in \calh_{code}$, we have that
	\begin{align}
	\begin{split}
	\braket{\Psi|u^\dagger \calp u \calo^\prime|\Phi} 
	&= \braket{\Psi|u^\dagger \calp  \tilde{\calo}^\prime u |\Phi}
	=\braket{\Psi|u^\dagger \tilde{\calo}^\prime \calp   u |\Phi}
	=\braket{\tilde{\calo}^{\prime \, \dagger} u\Psi|  \calp   u |\Phi} \\
	&=\braket{u \calo^{\prime \, \dagger} \Psi|  \calp   u |\Phi}
	=\braket{  \Psi| \calo^{\prime } u^\dagger \calp   u |\Phi}.
	\end{split}
	\end{align} 
	Hence, $u^\dagger \calp u \in M_{code}^{\prime \prime} = M_{code}$.
\end{proof}	

\subsection{Special case of bounded relative Tomita operator}
\label{sec:specialcase}

We will first prove Theorem \ref{thm:otherdirection} in the special case where the relative Tomita operators with respect to $M_{code}$ and $M_{phys}$, denoted respectively by $S_{\Psi|\Phi}^c$ and $S_{\Psi|\Phi}^p$, are bounded operators. In this special case, we do not have to keep track of their domains. The proof of the general case is similar, but technically more complicated.

For any $\calo \in M_{code}$,
\begin{equation} u S^c_{\Psi|\Phi} \calo \ket{\Psi} = u \calo^\dagger \ket{\Phi} =  \tilde{\calo}^\dagger u \ket{\Phi} =  S^p_{u\Psi|u\Phi}  \tilde{\calo} u \ket{\Psi} =  S^p_{u\Psi|u\Phi} u \calo  \ket{\Psi}, \end{equation}
Hence \begin{equation} \left(u S^c_{\Psi|\Phi}  -  S^p_{u\Psi|u\Phi} u\right) \calo  \ket{\Psi} = 0. \end{equation}
$\left(u S^c_{\Psi|\Phi}  -  S^p_{u\Psi|u\Phi} u\right)$ is a bounded operator that annihilates a dense subspace of $\calh_{code}$, since $\ket{\Psi}$ is cyclic with respect to $M_{code}$. It follows from the fact that the kernel of $\left(u S^c_{\Psi|\Phi}  -  S^p_{u\Psi|u\Phi} u\right)$ is closed that
\begin{equation} \label{eq:firstrelation} u S^c_{\Psi|\Phi}  =  S^p_{u\Psi|u\Phi} u.\end{equation}

Likewise, for any $\calp \in M_{phys}$,
\begin{align} 
u^\dagger S^p_{u \Psi | u \Phi} \calp \ket{u\Psi}  = u^\dagger \calp^\dagger  u\ket{\Phi} & = S^c_{\Psi | \Phi }u^\dagger \calp  \ket{u\Psi}, \\
\left(u^\dagger S^p_{u \Psi | u \Phi}   - S^c_{\Psi | \Phi }u^\dagger   \right) &\calp\ket{u\Psi} = 0. 
\end{align}
As $u\ket{\Psi}$ is cyclic with respect to $M_{phys}$ by assumption, we have that \begin{equation} 
\label{eq:secondrelation} 
u^\dagger S^p_{u \Psi | u \Phi}   = S^c_{\Psi | \Phi }u^\dagger, \quad   S^{p \, \dagger}_{u \Psi | u \Phi} u    = u S^{c \, \dagger}_{\Psi | \Phi }.
\end{equation}
Equations \eqref{eq:firstrelation} and \eqref{eq:secondrelation} imply that the subspace $\text{Im } u$ is mapped to itself under $S^{p}_{u \Psi | u \Phi}$ and $S^{p \, \dagger}_{u \Psi | u \Phi}$. Thus, the subspace $\text{Im } u$ is mapped to itself under $\Delta^{p}_{u \Psi | u \Phi}$. From the fact that $\Delta^{p}_{u \Psi | u \Phi}$ is self-adjoint and bounded, it follows that the subspace $(\text{Im } u)^\perp$ is mapped to itself under $\Delta^{p}_{u \Psi | u \Phi}$. Equations \eqref{eq:firstrelation} and \eqref{eq:secondrelation} also imply that \begin{equation}
\Delta^c_{\Psi|\Phi} = u^\dagger \Delta^p_{u\Psi|u\Phi}u. \label{eq:deltarelation}\end{equation}
Note that $\Delta^p_{u \Psi| u\Phi}$ and $\Delta^c_{\Psi|\Phi}$ are positive, self-adjoint, bounded operators. Thus, we may use the spectral theorem to study them. We will apply the spectral theorem to $(\Delta^p_{u \Psi| u\Phi})|_{\text{Im } u}$ and $(\Delta^p_{u \Psi| u\Phi})|_{(\text{Im } u)^\perp}$ separately.\footnote{$(\Delta^p_{u \Psi| u\Phi})|_{\text{Im } u}$ denotes the restriction of $\Delta^p_{u \Psi| u\Phi}$ to the closed subspace $\text{Im } u$.} We write
\begin{equation}
(\Delta^p_{u \Psi| u\Phi})|_{\text{Im } u} = \int_{\mathbb{R}} \lambda \, d (P^{\text{Im } u}_\lambda) , \
(\Delta^p_{u \Psi| u\Phi})|_{(\text{Im } u)^\perp} = \int_{\mathbb{R}} \lambda \, d (P^{(\text{Im } u)^\perp}_\lambda).
\end{equation}
For a Borel set $\Omega \subset \mathbb{R}$, the projections $P_\Omega^{\text{Im } u}$ and $P_\Omega^{(\text{Im } u)^\perp}$ commute with $u u^\dagger$ because $u u^\dagger$ is the projection onto $\text{Im } u$. The spectral decomposition of $\Delta^p_{u \Psi| u\Phi}$ is given by
\begin{equation}\Delta^p_{u \Psi| u\Phi} = \int_{\mathbb{R}} \lambda \, d (P^{p}_\lambda). \end{equation}
By the uniqueness of the spectral decomposition, we have that $P_\Omega^p = P_\Omega^{\text{Im } u} + P_\Omega^{(\text{Im } u)^\perp}$. Thus, $P_\Omega^p$ commutes with $u u^\dagger$. Let $\Omega_1$ and $\Omega_2$ be two Borel sets. Then
\begin{equation}
u^\dagger P_{\Omega_1}^p u  u^\dagger P_{\Omega_2}^p u = u^\dagger P_{\Omega_1}^p  P_{\Omega_2}^p u  .
\end{equation}
One can then check that the family of projections $u^\dagger P_{\Omega}^p u = u^\dagger P_{\Omega}^{\text{Im } u} u$ is a projection valued measure on $\calh_{code}$. We will now show that this is the projection valued measure associated with $\Delta^c_{\Psi|\Phi}$. From equation \eqref{eq:deltarelation}, it follows that for any $\ket{\Theta} \in \calh_{code}$, we have that
\begin{equation} \Delta^c_{\Psi | \Phi} \ket{\Theta} =  u^\dagger\Delta^p_{u \Psi | u \Phi} u\ket{ \Theta} = \int_\mathbb{R} \lambda \, d (u^\dagger P^p_\lambda  u \ket{  \Theta}). \end{equation}
By the uniqueness of the spectral decomposition of $\Delta^c_{\Psi | \Phi}$, we conclude that $u^\dagger P^p_\Omega u$ is the projection valued measure associated with $\Delta^c_{\Psi | \Phi}$. It follows that
\begin{align}
\begin{split} -\braket{\Psi|\log(\Delta^c_{\Psi | \Phi}) |\Psi} &= -\int_0^\infty \log(\lambda) \, d ( \braket{ \Psi| u^\dagger P^p_\lambda u |  \Psi}) \\&= - \int_0^\infty \log(\lambda) \, d ( \braket{ u\Psi| P^p_\lambda |u  \Psi}) = -\braket{u \Psi| \log(\Delta^p_{u \Psi | u \Phi}) |u\Psi}.
\end{split} \end{align}

The same logic can be applied to the commutant algebras $M_{code}^\prime$ and $M_{phys}^\prime$. Hence,
\begin{align} 
\cals_{\Psi|\Phi}(M_{code})= \cals_{u\Psi|u\Phi}(M_{phys}), \quad \cals_{\Psi|\Phi}(M_{code}^\prime)= \cals_{u\Psi|u\Phi}(M_{phys}^\prime) .
\end{align}

\subsection{General proof of Theorem \ref{thm:otherdirection}}
\label{sec:general}

\begin{lem}
\label{lem:reltomita}
Let $S^c_{\Psi | \Phi}$ denote the relative Tomita operator defined with respect to $M_{code}$. Let $S^p_{u\Psi | u\Phi}$ denote the relative Tomita operator defined with respect to $M_{phys}$. Let $S^{c \, \prime}_{\Psi | \Phi}$ and $S^{p \, \prime}_{u\Psi | u\Phi}$ denote the relative Tomita operators defined with respect to $M_{code}^\prime$ and $M_{phys}^\prime$. Then $u S^c_{\Psi | \Phi} = S^p_{u \Psi | u \Phi} u$ and $u S^{c \, \prime}_{\Psi | \Phi} = S^{p \, \prime}_{u \Psi | u \Phi} u$.
\end{lem}
\begin{proof}
	$D(S^c_{\Psi|\Phi})$ is given by all $\ket{x} \in \calh_{code}$ that may be written as
	\begin{equation}
	\ket{x} = \lim_{n \rightarrow \infty} \calo_n \ket{\Psi}
	\end{equation}
	for some sequence $\{\calo_n\} \in M_{code}$ such that the limit
	\begin{equation}
	\ket{y} := \lim_{n \rightarrow \infty} \calo_n^\dagger \ket{\Phi}
	\end{equation}
	exists. By definition, $S^c_{\Psi|\Phi} \ket{x} := \ket{y}$. Given $\ket{x}$ and $\ket{y}$ defined as above, it follows that
	\begin{equation}
	u\ket{x} = \lim_{n \rightarrow \infty} \tilde{\calo}_n u \ket{\Psi}, \quad
	u\ket{y} = \lim_{n \rightarrow \infty} \tilde{\calo}_n^\dagger u \ket{\Phi}.
	\end{equation}
	Hence, $u\ket{x} \in D(S^p_{u\Psi|u\Phi})$. It follows that for all $\ket{x} \in D(S^c_{\Psi|\Phi})$, 
	\[ u S^c_{\Psi|\Phi} \ket{x} =  S^p_{u\Psi|u\Phi} u \ket{x},  \]
	which means that $ S^p_{u\Psi|u\Phi} u$ is an extension of $u S^c_{\Psi|\Phi}$. To see that $S^p_{u\Psi|u\Phi} u$ is not a proper extension, suppose $\ket{w} \in D( S^p_{u\Psi|u\Phi} u) $. Then $u\ket{w} \in D(S^p_{u\Psi|u\Phi})$, meaning that there exists a sequence $\{\calp_n\} \in M_{phys}$ such that
	\begin{equation}
	u\ket{w} = \lim_{n \rightarrow \infty} \calp_n u \ket{\Psi},\ \text{and} \
	\lim_{n \rightarrow \infty} \calp_n^\dagger u \ket{\Psi}\ \text{exists.}
	\end{equation}
	We may also write $\ket{w} = \lim_{n \rightarrow \infty} u^\dagger \calp_n u \ket{\Psi}$. From Lemma \ref{lem:phystocode}, $u^\dagger \calp_n u \in M_{code}$. Hence, $\ket{w} \in D(S^c_{\Psi|\Phi})$; so we may write $uS^c_{\Psi|\Phi} =  S^p_{u\Psi | u\Phi} u$ as an operator equality because the operators on both sides have the same domain and act the same way on vectors in the domain. The same logic establishes that $u S^{c \, \prime}_{\Psi | \Phi} = S^{p \, \prime}_{u \Psi | u \Phi} u$.
\end{proof}

\begin{lem}
\label{lem:lastlemma}
	Let $\Delta^p_{u \Psi | u \Phi} := S^{p \, \dagger}_{u \Psi|u\Phi} S^p_{u \Psi|u\Phi}$ be the relative modular operator associated with $S^p_{u \Psi|u\Phi}$. Then,
	\begin{itemize}
			\item $\Delta^{p}_{u \Psi | u \Phi}$ maps the vector space $(\text{Im } u) \cap D(\Delta^p_{u \Psi|u\Phi})$ into $(\text{Im } u)$, and 
			$(\Delta^p_{u \Psi | u \Phi})|_{(\text{Im } u)}$
			 is densely defined on $(\text{Im } u)$.
		\item $\Delta^p_{u \Psi | u \Phi}$ maps the vector space $(\text{Im } u)^\perp \cap D(\Delta^p_{u \Psi|u\Phi})$ into $(\text{Im } u)^\perp$, and $(\Delta^p_{u \Psi | u \Phi})|_{(\text{Im } u)^\perp}$ is densely defined on $(\text{Im } u)^\perp$.
	\end{itemize}
\end{lem}
\begin{proof}
Let $\ket{x} \in D(S^p_{u\Psi | u \Phi})$. Then there exists a sequence $\{\calp_n\} \in M_{phys}$ such that
\begin{equation} \ket{x} = \lim_{n \rightarrow \infty} \calp_n \ket{u \Psi},\ \text{and} \ \lim_{n \rightarrow \infty} \calp_n^\dagger \ket{u \Phi}\ \text{exists.}\end{equation}
Then $u^\dagger \ket{x} \in D(S^c_{\Psi|\Phi})$.  We may write
\begin{equation} S^c_{\Psi|\Phi} u^\dagger \ket{x} = u^\dagger S^p_{u \Psi | u \Phi} \ket{x}. \label{eq:relation} \end{equation}
The fact that $u^\dagger \ket{x} \in D(S^c_{\Psi|\Phi})$ and Lemma \ref{lem:reltomita} together imply that $u u^\dagger \ket{x} \in D(S^p_{u\Psi | u\Phi})$.

We may uniquely decompose $\ket{x}$ into the sum
\begin{equation} \ket{x} = \ket{a} + \ket{b} \end{equation}
where $\ket{a} \in \text{Im } u$ and $\ket{b} \in (\text{Im } u)^\perp$. We know that $\ket{a} = u u^\dagger \ket{x} \in D(S^p_{u \Psi | u \Phi})$. As $D(S^p_{u \Psi | u \Phi})$ is a vector space, this implies that $\ket{b} \in D(S^p_{u \Psi | u \Phi})$.

It follows from the above that
\begin{equation} D(S^p_{u \Psi|u\Phi}) = \text{Im } u \cap D(S^p_{u \Psi|u\Phi}) \oplus (\text{Im } u)^\perp \cap D(S^p_{u \Psi|u\Phi}). \end{equation}

From equation \eqref{eq:relation},
\begin{equation} uu^\dagger S^p_{u \Psi | u \Phi} \ket{b} = 0, \end{equation}
which means that $S^p_{u \Psi | u \Phi}$ maps the vector space $(\text{Im } u)^\perp \cap D(S^p_{u \Psi|u\Phi}) \rightarrow (\text{Im } u)^\perp$.

From Lemma \ref{lem:reltomita} we may write, for all $\ket{x} \in D(S^p_{u\Psi | u \Phi})$,
\begin{equation} 
u S^c_{\Psi|\Phi} u^\dagger \ket{x}= S^p_{u \Psi | u \Phi} u u^\dagger \ket{x}.\end{equation}
It follows from $u u^\dagger \ket{x} = \ket{a}$ that
\begin{equation} 
u S^c_{\Psi|\Phi} u^\dagger \ket{x}= S^p_{u \Psi | u \Phi} \ket{a}.\end{equation}
It follows from $u^\dagger \ket{b} = 0$ that
\begin{equation} 
u S^c_{\Psi|\Phi} u^\dagger \ket{a} = S^p_{u \Psi | u \Phi}  \ket{a},\end{equation}
which means that $S^p_{u \Psi | u \Phi}$ maps the vector space $(\text{Im } u) \cap D(S^p_{u \Psi|u\Phi}) \rightarrow (\text{Im } u)$.

We will now show that $(\text{Im } u) \cap D(S^p_{u \Psi|u\Phi})$ is dense in $(\text{Im } u)$. Given any $\ket{A} \in (\text{Im } u)$, choose $\ket{X} \in \calh_{phys}$ such that $uu^\dagger \ket{X} = \ket{A}$. Next, choose a sequence $\{\ket{x_n}\} \in D(S^p_{u \Psi|u\Phi})$ that converges to $\ket{X}$. We then have that $\lim_{n \rightarrow \infty} uu^\dagger \ket{x_n} =\ket{A}$. Since $\ket{x_n} \in D(S^p_{u \Psi|u\Phi})$, we know from earlier that $uu^\dagger \ket{x_n} \in D(S^p_{u \Psi|u\Phi})$. Hence, $(\text{Im } u) \cap D(S^p_{u \Psi|u\Phi})$ is dense in $(\text{Im } u)$. The same logic shows that $(\text{Im } u)^\perp \cap D(S^p_{u \Psi|u\Phi})$ is dense in $(\text{Im } u)^\perp$.

Furthermore, $(S^p_{u\Psi | u \Phi})|_{(\text{Im } u)}$ is a closed operator because $(\text{Im } u)$ is a closed subspace.

We can apply all of the above logic to the commutant algebras.  To summarize,
\begin{itemize}
	\item $S^p_{u \Psi | u \Phi}$ maps the vector space $(\text{Im } u)^\perp \cap D(S^p_{u \Psi|u\Phi})\rightarrow(\text{Im } u)^\perp$, and $(S^p_{u \Psi | u \Phi})|_{(\text{Im } u)^\perp}$ is closed and densely defined on $(\text{Im } u)^\perp$.
	\item $S^p_{u \Psi | u \Phi}$ maps the vector space $(\text{Im } u) \cap D(S^p_{u \Psi|u\Phi})\rightarrow(\text{Im } u)$, and $(S^p_{u \Psi | u \Phi})|_{(\text{Im } u)}$ is closed and densely defined on $(\text{Im } u)$.
	\item $S^{p \, \prime}_{u \Psi | u \Phi}$ maps the vector space $(\text{Im } u)^\perp \cap D(S^{p \, \prime}_{u \Psi|u\Phi})\rightarrow(\text{Im } u)^\perp$, and $(S^{p \, \prime}_{u \Psi | u \Phi})|_{(\text{Im } u)^\perp}$ is closed and densely defined on $(\text{Im } u)^\perp$.
	\item $S^{p \, \prime}_{u \Psi | u \Phi}$ maps the vector space $(\text{Im } u) \cap D(S^{p \, \prime}_{u \Psi|u\Phi})\rightarrow(\text{Im } u)$, and $(S^{p \, \prime}_{u \Psi | u \Phi})|_{(\text{Im } u)}$ is closed and densely defined on $(\text{Im } u)$.
\end{itemize}
It directly follows that the above statements also hold for the adjoints $S^{p \, \dagger}_{u \Psi | u \Phi}$ and $S^{p \, \prime \, \dagger}_{u \Psi | u \Phi}$.
Recall that $\Delta^p_{u \Psi | u \Phi} = S^{p \, \dagger}_{u \Psi | u \Phi} S^p_{u \Psi | u \Phi}$. We may compute $(\Delta^p_{u \Psi | u \Phi})|_{(\text{Im } u)}$ from $(S^p_{u \Psi | u \Phi})|_{(\text{Im } u)}$ and $(\Delta^p_{u \Psi | u \Phi})|_{(\text{Im } u)^\perp}$ from $(S^p_{u \Psi | u \Phi})|_{(\text{Im } u)^\perp}$.
In particular, $(\Delta^p_{u \Psi | u \Phi})|_{(\text{Im } u)}$ is given by 
\begin{equation}
(\Delta^p_{u \Psi | u \Phi})|_{(\text{Im } u)} = (S^{p \, \dagger}_{u \Psi | u \Phi}|_{(\text{Im } u)})( S^p_{u \Psi | u \Phi}|_{(\text{Im } u)}) = (S^{p }_{u \Psi | u \Phi}|_{(\text{Im } u)})^\dagger (S^p_{u \Psi | u \Phi}|_{(\text{Im } u)}).
\end{equation} 
It follows that $(\Delta^p_{u \Psi | u \Phi})|_{(\text{Im } u)}$ is densely defined and self-adjoint on $(\text{Im }u)$. The same logic can be applied to $(\Delta^p_{u \Psi | u \Phi})|_{(\text{Im } u)^\perp}$.
\end{proof}

Having established Lemmas \ref{lem:firstlemma} to \ref{lem:lastlemma}, we can now prove Theorem \ref{thm:otherdirection}, which shows that entanglement wedge reconstruction implies the equivalence of bulk and boundary relative entropies.

\begin{thm:otherdirection}
Let $u : \calh_{code}\rightarrow \calh_{phys}$ be an isometry between two Hilbert spaces. Let $M_{code}$ and $M_{phys}$ be von Neumann algebras on $\calh_{code}$ and $\calh_{phys}$ respectively. Let $M^\prime_{code}$ and $M^\prime_{phys}$ respectively be the commutants of $M_{code}$ and $M_{phys}$. 

\noindent Suppose that

\begin{itemize}
	\item There exists some state $\ket{\Omega} \in \calh_{code}$ such that $u\ket{\Omega} \in \calh_{phys}$ is cyclic and separating with respect to $M_{phys}$. 
	\item $\forall \calo \in M_{code}\ \forall \calo^\prime \in M_{code}^\prime, \quad 
	\exists\tilde{\calo} \in M_{phys}\ \exists \tilde{\calo}^\prime \in M_{phys}^\prime$ such that
	\begin{align} \nonumber
	\begin{split}
	\forall \ket{\Theta} \in \calh_{code} \quad 
	\begin{cases}
	u \calo \ket{\Theta} =  \tilde{\calo} u \ket{\Theta}, \quad
	&u \calo^\prime \ket{\Theta} =  \tilde{\calo}^\prime u \ket{\Theta}, \\
	u \calo^\dagger \ket{\Theta} =  \tilde{\calo}^\dagger u \ket{\Theta}, \quad
	&u \calo^{\prime \dagger} \ket{\Theta} = \tilde{\calo}^{\prime\dagger} u \ket{\Theta}.
	\end{cases}
	\end{split}
	\end{align}
\end{itemize}

\noindent Then, for any $\ket{\Psi}$, $\ket{\Phi} \in \calh_{code}$ with $\ket{\Psi}$ cyclic and separating with respect to $M_{code}$, 
\begin{itemize}
	\item $u \ket{\Psi}$ is cyclic and separating with respect to $M_{phys}$ and $M_{phys}^\prime$,
	\item $\cals_{\Psi|\Phi}(M_{code})= \cals_{u\Psi|u\Phi}(M_{phys}), \quad \cals_{\Psi|\Phi}(M_{code}^\prime)= \cals_{u\Psi|u\Phi}(M_{phys}^\prime),$
\end{itemize}
where $\cals_{\Psi|\Phi}(M)$ is the relative entropy.
\end{thm:otherdirection}

\begin{proof}
	$\Delta^p_{u \Psi| u\Phi}$ and $\Delta^c_{\Psi|\Phi}$ are positive, densely defined, self-adjoint operators that are generically unbounded. Thus, we may use the spectral theorem to study them. We will apply the spectral theorem to $(\Delta^p_{u \Psi| u\Phi})|_{\text{Im } u}$ and $(\Delta^p_{u \Psi| u\Phi})|_{(\text{Im } u)^\perp}$ separately. We write
	\begin{equation}
(\Delta^p_{u \Psi| u\Phi})|_{\text{Im } u} = \int_{\mathbb{R}} \lambda \, d (P^{\text{Im } u}_\lambda) , \
(\Delta^p_{u \Psi| u\Phi})|_{(\text{Im } u)^\perp} = \int_{\mathbb{R}} \lambda \, d (P^{(\text{Im } u)^\perp}_\lambda).
\end{equation}
For a Borel set $\Omega \subset \mathbb{R}$, the projections $P_\Omega^{\text{Im } u}$ and $P_\Omega^{(\text{Im } u)^\perp}$ commute with $u u^\dagger$ because $u u^\dagger$ is the projection onto $\text{Im } u$. The spectral decomposition of $\Delta^p_{u \Psi| u\Phi}$ is given by
\begin{equation}\Delta^p_{u \Psi| u\Phi} = \int_{\mathbb{R}} \lambda \, d (P^{p}_\lambda). \end{equation}
By the uniqueness of the spectral decomposition, we have that $P_\Omega^p = P_\Omega^{\text{Im } u} + P_\Omega^{(\text{Im } u)^\perp}$. Thus, $P_\Omega^p$ commutes with $u u^\dagger$. Let $\Omega_1$ and $\Omega_2$ be two Borel sets. Then
\begin{equation}
u^\dagger P_{\Omega_1}^p u  u^\dagger P_{\Omega_2}^p u = u^\dagger P_{\Omega_1}^p  P_{\Omega_2}^p u  .
\end{equation}
One can then check that the family of projections $u^\dagger P_{\Omega}^p u = u^\dagger P_{\Omega}^{\text{Im } u} u$ is a projection valued measure on $\calh_{code}$. We will now show that this is the projection valued measure associated with $\Delta^c_{\Psi|\Phi}$. From Lemma \ref{lem:reltomita} we have that
\begin{equation}
uS^{c }_{\Psi|\Phi} u^\dagger =  S^p_{u\Psi | u\Phi} u u^\dagger = (S^p_{u\Psi | u\Phi})|_{(\text{Im } u)} . 
\end{equation}
We may take the adjoint of the above equation to obtain
\begin{equation}
uS^{c \, \dagger }_{\Psi|\Phi} u^\dagger = (S^{p \, \dagger}_{u\Psi | u\Phi})|_{(\text{Im } u)}, 
\end{equation}
from which it follows that
\begin{equation}
\label{eq:modularoperators}
u\Delta^{c }_{\Psi|\Phi}u^\dagger = (\Delta^p_{u\Psi | u\Phi})|_{(\text{Im } u)}, \quad \Delta^{c }_{\Psi|\Phi} = u^\dagger \Delta^p_{u\Psi | u\Phi}u .
\end{equation}
For any $\ket{\Theta} \in D(\Delta^c_{\Psi | \Phi})$, we have that
\begin{equation} \Delta^c_{\Psi | \Phi} \ket{\Theta} =  u^\dagger\Delta^p_{u \Psi | u \Phi} u\ket{ \Theta} = \int_\mathbb{R} \lambda \, d (u^\dagger P^p_\lambda  u \ket{  \Theta}). \end{equation}
By the uniqueness of the spectral decomposition of $\Delta^c_{\Psi | \Phi}$, we conclude that $u^\dagger P^p_\Omega u$ is the projection valued measure associated with $\Delta^c_{\Psi | \Phi}$.

It follows that
\begin{align}
\begin{split} -\braket{\Psi|\log(\Delta^c_{\Psi | \Phi}) |\Psi} &= -\int_0^\infty \log(\lambda) \, d ( \braket{ \Psi| u^\dagger P^p_\lambda u |  \Psi}) \\&= - \int_0^\infty \log(\lambda) \, d ( \braket{ u\Psi| P^p_\lambda |u  \Psi}) = -\braket{u \Psi| \log(\Delta^p_{u \Psi | u \Phi}) |u\Psi}.
\end{split} \end{align}

The same logic can be applied to the commutant algebras $M_{code}^\prime$ and $M_{phys}^\prime$. Hence,
\begin{align} 
\cals_{\Psi|\Phi}(M_{code})= \cals_{u\Psi|u\Phi}(M_{phys}), \quad \cals_{\Psi|\Phi}(M_{code}^\prime)= \cals_{u\Psi|u\Phi}(M_{phys}^\prime) .
\end{align}
\end{proof}

\section{Proof of Theorem \ref{thm:maintheorem}}
\label{sec:converse}
\begin{thm:maintheorem}
	Let $u : \calh_{code}\rightarrow \calh_{phys}$ be an isometry between two Hilbert spaces. Let $M_{code}$ and $M_{phys}$ be von Neumann algebras on $\calh_{code}$ and $\calh_{phys}$ respectively. Let $M^\prime_{code}$ and $M^\prime_{phys}$ respectively be the commutants of $M_{code}$ and $M_{phys}$. Suppose that the set of cyclic and separating vectors with respect to $M_{code}$ is dense in $\calh_{code}$. Also suppose that if $\ket{\Psi} \in \calh_{code}$ is cyclic and separating with respect to $M_{code}$, then $u \ket{\Psi}$ is cyclic and separating with respect to $M_{phys}$. Then the following two statements are equivalent:

\begin{description}
	\item[ 1. Bulk reconstruction ]
	
	$\forall \calo \in M_{code}\ \forall \calo^\prime \in M_{code}^\prime, \quad 
	\exists\tilde{\calo} \in M_{phys}\ \exists \tilde{\calo}^\prime \in M_{phys}^\prime$ such that
	\begin{align} \nonumber
	\begin{split}
	\forall \ket{\Theta} \in \calh_{code} \quad 
	\begin{cases}
	u \calo \ket{\Theta} =  \tilde{\calo} u \ket{\Theta}, \quad
	&u \calo^\prime \ket{\Theta} =  \tilde{\calo}^\prime u \ket{\Theta}, \\
	u \calo^\dagger \ket{\Theta} =  \tilde{\calo}^\dagger u \ket{\Theta}, \quad
	&u \calo^{\prime \dagger} \ket{\Theta} = \tilde{\calo}^{\prime\dagger} u \ket{\Theta}.
	\end{cases}
	\end{split}
	\end{align}
	
	\item[ 2. Relative entropy equals bulk relative entropy]
	For any $\ket{\Psi}$, $\ket{\Phi} \in \calh_{code}$ with $\ket{\Psi}$ cyclic and separating with respect to $M_{code}$, \[\cals_{\Psi|\Phi}(M_{code})= \cals_{u\Psi|u\Phi}(M_{phys}), \text{and} \  \cals_{\Psi|\Phi}(M_{code}^\prime)= \cals_{u\Psi|u\Phi}(M_{phys}^\prime),\] where $\cals_{\Psi|\Phi}(M)$ is the relative entropy.
\end{description}

\end{thm:maintheorem}

\begin{proof}
Given the proof of Theorem \ref{thm:otherdirection}, we only need to show that statement 2 implies statement 1. Let $\ket{\Phi} \in \calh_{code}$ be cyclic and separating with respect to $M_{code}$, and let $U \in M_{code}$ and $U^\prime \in M_{code}^\prime$ be unitary operators. We can easily see that
\begin{equation}
0 = \cals_{\Phi | U^\prime \Phi}(M_{code}) = \cals_{u\Phi | uU^\prime \Phi}(M_{phys}).
\end{equation}
Due to Theorem \ref{thm:saturation}, this implies that
\begin{equation}
\Delta^p_{u \Phi | u U^\prime \Phi} \ket{u \Phi} = \ket{u \Phi},
\end{equation}
where $\Delta^p_{u \Phi | u U^\prime \Phi} = S^{p \, \dagger}_{u \Phi | u U^\prime \Phi} S^p_{u \Phi | u U^\prime \Phi}$ and $S^p_{u \Phi | u U^\prime \Phi}$ is the relative modular operator defined with respect to $M_{phys}$. It follows that for any $\calp \in M_{phys}$,
\begin{equation}
\label{eq:matrixelofp}
\braket{ u U^\prime \Phi |  \calp u U^\prime \Phi }
=
\braket{S^p_{u \Phi | u U^\prime \Phi} u \Phi | S^{p}_{u \Phi | u U^\prime \Phi} \calp^\dagger u  \Phi}
=
\braket{\calp^\dagger u  \Phi|S^{p \, \dagger}_{u \Phi | u U^\prime \Phi} S^p_{u \Phi | u U^\prime \Phi} u \Phi}  = \braket{u \Phi |\calp | u \Phi}.
\end{equation}
This implies that
\begin{equation}
\braket{    \Phi   | U^{\prime \, \dagger} u^\dagger \calp u U^\prime - u^\dagger \calp u|\Phi}  = 0.
\end{equation}
	We now use the assumption that cyclic and separating vectors with respect to $M_{code}$ are dense in $\calh_{code}$. For any $\ket{\Psi} \in \calh_{code}$, choose a sequence $\{\ket{\Phi_n}\} \in \calh_{code}$ such that each $\ket{\Phi_n}$ is cyclic and separating with respect to $M_{code}$, and $\ket{\Psi} = \lim_{n \rightarrow \infty} \ket{\Phi_n}$. Then,
\begin{equation}
\braket{    \Psi   | U^{\prime \, \dagger} u^\dagger \calp u U^\prime - u^\dagger \calp u|\Psi} = \lim_{n \rightarrow \infty} \braket{    \Phi_n   | U^{\prime \, \dagger} u^\dagger \calp u U^\prime - u^\dagger \calp u|\Phi_n} = 0.
\end{equation}
Hence, this implies that the  operators that are measured in the limit itself is zero, i.e. $U^{\prime \, \dagger} u^\dagger \calp u U^\prime - u^\dagger \calp u = 0$. This then gives the following identity involving the isometry $u$, an arbitrary operator $\calp \in M_{phys}$, and a unitary operator $U^\prime \in M_{code}^{\prime}$:
\begin{equation}
u^\dagger \calp u U^\prime = U^{\prime } u^\dagger \calp u .
\end{equation}
The same logic can be applied to the commutant algebras; thus, for any $\calp^\prime \in M_{phys}^\prime$, $U \in M_{code}$ with $U$ unitary, we have a similar relation:
\begin{equation}
\label{eq:commuting}
u^\dagger \calp^\prime u U = U u^\dagger \calp^\prime u.
\end{equation}

Another consequence of equation \eqref{eq:matrixelofp} is that for any $\calp_1,\calp_2 \in M_{phys}$, we have that
\begin{equation} \label{eq:unitary} \braket{\calp_1 u U^\prime \Phi | \calp_2 u U^\prime \Phi} = \braket{ \calp_1 u \Phi | \calp_2 u \Phi}. \end{equation}
Naturally, we define a linear map $X^{\prime \, \Phi \, U^\prime} : \calh_{phys}\rightarrow \calh_{phys}$. We define $X^{\prime \, \Phi \, U^\prime}$ by
\begin{equation}
\label{eq:defXp}
X^{\prime \, \Phi \, U^\prime}  \calp u \ket{\Phi} := \calp u U^\prime \ket{\Phi} \quad \forall \calp \in M_{phys}.
\end{equation}
Then we see that $X^{\prime \, \Phi \, U^\prime}$ is densely defined. From equation \eqref{eq:unitary}, we see that $X^{\prime \, \Phi \, U^\prime}$ preserves the norm of all vectors in its domain. Hence, $X^{\prime \, \Phi \, U^\prime}$ may be uniquely extended to a bounded operator, which is unitary. By definition, $X^{\prime \, \Phi \, U^\prime}$ commutes with all operators in $M_{phys}$; hence, we deduce that $X^{\prime \, \Phi \, U^\prime} \in M_{phys}^\prime$. (The superscripts on $X^{\prime \, \Phi \, U^\prime}$ remind us that it depends on the choice of $\ket{\Phi}$ and $U^\prime$ and that it is in the commutant of $M_{phys}$.)

Next, we use equations \eqref{eq:commuting} and \eqref{eq:defXp} with $\calp^\prime = X^{\prime \, \Phi \, U^\prime}$. We find that
\begin{equation}
u^\dagger X^{\prime \, \Phi \, U^\prime} u U \ket{\Phi} = U u^\dagger X^{\prime \, \Phi \, U^\prime} u \ket{\Phi} = U u^\dagger u U^\prime \ket{\Phi} = U U^\prime \ket{\Phi} = U^\prime U \ket{\Phi}.
\end{equation}
The first equality follows from equation \eqref{eq:commuting}, the second equality follows from \eqref{eq:defXp}, the third equality follows from the fact that $u^\dagger u$ is the identity on $\calh_{code}$, and the last equality follows because $U \in M_{code}$ and $U^\prime \in M_{code}^\prime$. Recall that $U$ is an arbitrary unitary operator in $M_{code}$. We now need Theorem \ref{thm:fourunitaries}, which states that any operator in $M_{code}$ may be written as a linear combination of four unitary operators in $M_{code}$ \cite{Jones-vNalg}. The above equation implies that for any $\calo \in M_{code}$, we have that
\begin{equation}
(u^\dagger X^{\prime \, \Phi \, U^\prime} u - U^\prime)\calo \ket{\Phi} = 0.
\end{equation}
Note that $(u^\dagger X^{\prime \, \Phi \, U^\prime} u - U^\prime)$ is a bounded operator, so its kernel is closed. Recall that $\ket{\Phi}$ is cyclic with respect to $M_{code}$. Since any vector in the Hilbert space may be written as $\lim_{n \rightarrow \infty} \calo_n \ket{\Phi}$ for some sequence of operators $\{\calo_n\} \in M_{code}$, it follows that $(u^\dagger X^{\prime \, \Phi \, U^\prime} u - U^\prime)$ annihilates every vector in $\calh_{code}$. In other words,
\begin{equation}
\label{eq:operatormap}
u^\dagger X^{\prime \, \Phi \, U^\prime} u = U^\prime.
\end{equation}
Choose an arbitrary $\ket{\Psi} \in \calh_{code}$ with $\braket{\Psi|\Psi} = 1$. We may uniquely write $X^{\prime \, \Phi \, U^\prime} u \ket{\Psi}$ as
\begin{equation}
X^{\prime \, \Phi \, U^\prime} u\ket{\Psi} = \ket{a} + \ket{b},
\end{equation}
where $\ket{a} \in \text{Im } u$, and $\ket{b} \in (\text{Im u})^\perp$. Note that $X^{\prime \, \Phi \, U^\prime}$ is unitary; hence, we can decompose as
\begin{equation}
\braket{u \Psi | X^{\prime \, \Phi \, U^\prime \, \dagger} X^{\prime \, \Phi \, U^\prime} | u \Psi}   = 1 = \braket{a|a} + \braket{b|b}.
\end{equation}
Next, note that
\begin{equation}
u^\dagger \ket{a} = u^\dagger (\ket{a} + \ket{b})  
=u^\dagger X^{\prime \, \Phi \, U^\prime} u \ket{\Psi} = U^\prime \ket{\Psi}.
\end{equation}
Hence, 
\begin{equation}
\braket{a |a} = \braket{u^\dagger a | u^\dagger a} = \braket{U^\prime \Psi| U^\prime \Psi} = 1.
\end{equation}
This implies that $\braket{b|b} = 0$; hence $\ket{b} = 0$. Hence, $X^{\prime \Phi U^\prime}$ maps the vector space $\text{Im u}$ to itself. We may then use equation \eqref{eq:operatormap} to find that
\begin{equation}
\label{eq:bulktoboundaryequivalence}
X^{\prime \, \Phi \, U^\prime} u = u u^\dagger X^{\prime \, \Phi \, U^\prime} u = u U^\prime.
\end{equation}

Next, we define a linear map $X^{\prime \, (U^\prime \Phi) \, (U^{\prime \, \dagger})} : \calh_{phys}\rightarrow \calh_{phys}$. We define $X^{\prime \, (U^\prime \Phi) \, (U^{\prime \, \dagger})}$ by
\begin{equation}
X^{\prime \, (U^\prime \Phi) \, (U^{\prime \, \dagger})}  \calp u  U^\prime \ket{\Phi} := \calp u  \ket{\Phi} \quad \forall \calp \in M_{phys}.
\end{equation}
It is easy to see that $U^\prime \ket{\Phi}$ is cyclic and separating with respect to $M_{code}$ given that $\ket{\Phi}$ is cyclic and separating with respect to $M_{code}$ and that $U^\prime \in M_{code}^\prime$ is unitary. It follows that $X^{\prime \, (U^\prime \Phi) \, (U^{\prime \, \dagger})}$ is densely defined and uniquely extends to a bounded operator, which is unitary. Since equation \eqref{eq:bulktoboundaryequivalence} is true for any $\ket{\Phi} \in \calh_{code}$ that is cyclic and separating with respect to $M_{code}$ and any unitary $U^\prime \in M_{code}^\prime$,
\begin{equation}
X^{\prime \, (U^\prime \Phi) \, (U^{\prime \, \dagger})} u  = u U^{\prime \, \dagger}.
\end{equation}
This relation can be used to see that for any $\calp \in M_{phys}$,
\begin{equation}
X^{\prime \, (U^\prime \Phi) \, (U^{\prime \, \dagger})}  X^{\prime \, \Phi \, U^\prime}  \calp u \ket{\Phi} = \calp u  \ket{\Phi}.
\end{equation}
Thus, we deduce that the two operators we defined are adjoints of each other:
\begin{equation}
(X^{\prime \, \Phi \, U^\prime} )^\dagger = X^{\prime \, (U^\prime \Phi) \, (U^{\prime \, \dagger})}.
\end{equation}
We have thus shown that for every unitary operator $U^\prime \in M_{code}^\prime$, there exists a unitary operator $X^\prime \in M_{phys}^\prime$ such that
\begin{equation}
X^{\prime } u = u U^\prime, \ \text{and} \ X^{\prime \, \dagger} u = u U^{\prime \dagger}.
\end{equation}
The same logic applies to show that for every unitary operator $U \in M_{code}$, there exists a unitary operator $X \in M_{phys}$ such that
\begin{equation}
X u = u U, \ \text{and} \ X^{ \dagger} u = u U^{\dagger}.
\end{equation}
We conclude the proof by noting that any operator in a von Neumann algebra $M$ may be written as a linear combination of four unitary operators in $M$ (Theorem \ref{thm:fourunitaries}).
\end{proof}

Our proof provides an explicit formula for reconstructing an operator in $M_{code}$ as an operator in $M_{phys}$. Given $\calo \in M_{code}$, we define the operator $\tilde{\calo} \in M_{phys}$ by 
\begin{equation}
\label{eq:reconstructionformula}
\tilde{\calo} \calp^\prime u \ket{\Phi} := \calp^\prime u \calo \ket{\Phi} \quad \forall \calp^\prime \in M_{phys}^\prime,  
\end{equation}
where $\ket{\Phi} \in \calh_{code}$ is a fiducial state that is cyclic and separating with respect to $M_{code}$ and $M_{code}^\prime$. This formula follows from writing $\calo$ as a linear combination of four unitary operators in $M_{code}$ and using equation \eqref{eq:defXp} on each unitary operator. The arguments in our proof then establish that $\tilde{\calo} u = u \calo$. Note that $\tilde{\calo}$ does not depend on the choice of the fiducial state $\ket{\Phi}$. To see this, we define $\tilde{\calo}_\star \in M_{code}$ by
\begin{equation}
\tilde{\calo}_\star \calp^\prime u \ket{\Phi_\star} := \calp^\prime u \calo \ket{\Phi_\star} \quad \forall \calp^\prime \in M_{phys}^\prime,  
\end{equation}
where $\ket{\Phi_\star} \in \calh_{code}$ is a different fiducial state. Since $\tilde{\calo} u \ket{\Phi_\star} =  u \calo \ket{\Phi_\star}$, it follows that
\begin{equation}
\tilde{\calo} \calp^\prime u \ket{\Phi_\star} =  \calp^\prime  \tilde{\calo} u \ket{\Phi_\star}
=  \calp^\prime   u \calo \ket{\Phi_\star} = \tilde{\calo}_\star \calp^\prime u \ket{\Phi_\star}  \quad \forall \calp^\prime \in M_{phys}^\prime. 
\end{equation}
Hence, $\tilde{\calo}$ and  $\tilde{\calo}_\star$ are equal because they are both bounded operators that act the same way on a dense subspace of $\calh_{code}$.

\section{Discussion}
\label{sec:discussion}

In this section, we discuss the physical implications of Theorem \ref{thm:maintheorem}. In particular, we explain in physical settings the validity of the technical assumptions of the theorem. In Section \ref{sec:vnalgebras}, we motivate our use of von Neumann algebras by explaining how they arise in quantum field theory, with an approach inspired by \cite{Haag}. In Section \ref{sec:approx}, we summarize reasons why Theorem \ref{thm:maintheorem} is only approximately applicable to quantum gravity. In Section \ref{sec:rs}, we summarize the Reeh--Schlieder theorem. In Section \ref{sec:rsapplied}, we use the Reeh--Schlieder theorem to physically motivate the assumptions of Theorem \ref{thm:maintheorem}. In Section \ref{sec:implications}, we compare Theorem \ref{thm:maintheorem} with previous work on finite-dimensional error correction \cite{Harlow:2016fse}. %{\color{red} While our proofs of Theorems \ref{thm:maintheorem}  and \ref{thm:otherdirection} are mathematically rigorous, this section is not.}

\subsection{Von Neumann algebras in quantum field theory}
\label{sec:vnalgebras}

 Quantum field theories are characterized by algebras of operators acting on a Hilbert space $\calh$. For every open region in spacetime, there is an associated algebra \cite{Haag}. We will assume that there is a unique ground state $\ket{\Omega} \in \calh$. The closure of the set of states obtained by acting on $\ket{\Omega}$ with all operators in the algebra associated with the entire spacetime is defined to be the vacuum superselection sector, $\calh_0$. By definition, each superselection sector of the theory is an invariant subspace of this algebra.

Theories with lagrangian descriptions have a notion of an elementary field. Given an open region of spacetime $\calu$, we can define an associated operator algebra $\cala(\calu)$ by smearing the elementary fields with functions supported only in $\calu$.\footnote{ Assuming that the time-slice axiom \cite{Haag} holds, $\cala(\calu)$ should really be associated with the domain of dependence of $\calu$, as operators in the domain of dependence are related to operators in $\calu$ via an equation of motion. Note that the time-slice axiom does not hold for generalized free fields \cite{DuetschRehren}, which we consider in Section \ref{sec:approx}.} The operator algebra $\cala(\calu)$ generically contains unbounded operators. Given $\cala(\calu)$, we may obtain a von Neumann algebra $M(\calu)$, which only consists of bounded operators, as follows \cite{Haag}. For every unbounded operator (which we assume to be closed) in $\cala(\calu)$, we may perform a polar decomposition to obtain a partial isometry and a self-adjoint positive operator, which is canonically associated with a set of projections by the spectral theorem. The von Neumann algebra $M(\calu)$ is generated by the set of all spectral projections and partial isometries associated with the operators in $\cala(\calu)$.\footnote{If a subalgebra $\cals$ of bounded operators contains the identity and is closed under hermitian conjugation, then its double commutant, $S^{\prime \prime}$, is the von Neumann algebra generated by $\cals$. Von Neumann algebras are naturally associated with causally complete subregions \cite{Witten:2018zxz,Jones-vNalg}.} We assume that the operators in $\cala(\calu)$ may be approximated by operators in $M(\calu)$. As shown in \cite{Witten:2018zxz}, the Reeh--Schlieder theorem implies that states with bounded energy-momentum are cyclic with respect to $\cala(\calu)$ for any open subregion of spacetime $\calu$. We assume that this is also true for $M(\calu)$.

\subsection{Approximate entanglement wedge reconstruction}
\label{sec:approx}

Throughout the paper, we have used von Neumann algebras to denote subregions in the bulk and the boundary. In AdS/CFT, the boundary theory is a quantum field theory, so the  discussion in Section \ref{sec:vnalgebras} directly applies. However, the bulk theory is a theory of quantum gravity (string theory). For states with a semi-classical bulk dual, the bulk theory may be effectively described using quantum field theory on an asymptotically AdS background that might contain black holes. The applicability of quantum field theory motivates us to use von Neumann algebras to describe operators associated with covariantly defined subregions in the bulk, like the entanglement wedge of a boundary subregion.\footnote{Associating a set of operators with a subregion in the bulk is highly nontrivial due to nonlocal effects in the bulk \cite{ghosh2017}. This is addressed in \cite{ghosh2018}, which studies information measures for sets of operators that are not closed under multiplication. We do not consider this subtlety in our analysis.} Since entanglement wedges are causally complete, they naturally have an associated von Neumann algebra.

Since the long-distance bulk physics is only approximately described by quantum field theory, we need a generalization of Theorem \ref{thm:maintheorem} that relates the approximate bulk reconstruction to the approximate equivalence of relative entropies between the boundary and the bulk. We want to note that our formulation of bulk reconstruction in Theorem \ref{thm:maintheorem} is exact in the sense that correlation functions of operators in $M_{code}$ exactly equal correlation functions computed on the boundary with the corresponding operators in $M_{phys}$.

To be more precise, Theorem \ref{thm:maintheorem} is only valid for certain choices of the code subspace. If the code subspace consists of states with semi-classically distinct geometries, it is not clear how von Neumann algebras can be associated with subregions in a state independent way. For Theorem \ref{thm:maintheorem} to be relevant, we could choose $\calh_{code}$ to be a subspace describing long wavelength modes in quantum field theory on a fixed background and the entanglement wedge to be the classical minimal area surface corresponding to a boundary subregion.  
To order $G_N^0$, the bulk dual of entanglement entropy is given by the bulk entanglement entropy of the entanglement wedge plus a local integral on the minimal area surface \cite{FLM}. This was used to relate the bulk and boundary modular hamiltonians \cite{Jafferis:2015del}. Since the bulk and boundary modular hamiltonians only differ by operators localized on the minimal surface, the bulk and boundary relative entropies are equivalent up to $\calo(G_N)$ corrections \cite{Jafferis:2015del}. The bulk dual of relative entropy beyond order $G_N^0$ involves bulk modular hamiltonians evaluated with respect to different bulk surfaces \cite{Dong2017}.\footnote{It will be interesting to generalize equation (5.4) in \cite{Dong2017} to an expression that uses infinite-dimensional von Neumann algebras.} Since the formula for the bulk dual of relative entropy in Theorem \ref{thm:maintheorem} is only valid to order $G_N^0$, the two main statements in Theorem \ref{thm:maintheorem} can only be true in quantum gravity in an approximate sense. Theorem 4 of \cite{Cotler} proves that in the case of finite-dimensional von Neumann algebras, the approximate equivalence of bulk and boundary relative entropies implies approximate bulk reconstruction. Furthermore, \cite{alphabits} proves that entanglement wedge reconstruction can be exact to all orders in perturbation theory.\footnote{However,  in certain contexts, the entanglement wedge reconstruction  proposal  must be  nonperturbatively  approximate (see \cite{Kelly,alphabits}).}

It is possible for both statements in Theorem \ref{thm:maintheorem} to be exactly true in the limit $G_N \rightarrow 0$. In this case, the AdS/CFT duality relates a $(d+1)$-dimensional quantum field theory in AdS and a $d$-dimensional generalized free field theory, for which all connected $n$-point correlation functions vanish when $n \ge 3$.\footnote{The fact that all correlation functions may be expressed in terms of two-point functions arises from large-N factorization in the boundary CFT.} We may set $\calh_{code} = \calh_{phys}$ because every state in the boundary theory has a geometric dual. The case where the bulk theory is a free scalar is studied in \cite{DuetschRehren}. The authors of \cite{DuetschRehren} work in Poincar\'e coordinates, which has $d$-dimensional Minkowski space as its conformal boundary. They argue that in the boundary generalized free field theory, the algebra associated with the domain of dependence of any ball-shaped region in a spatial slice of Minkowski space is equal to the algebra associated with the causal wedge in the bulk.\footnote{This statement is also true for conformal transformations of such regions. For these boundary regions, the causal wedge is the same as the entanglement wedge \cite{Rangamani}.} This statement is expressed in equation (5.7) of \cite{DuetschRehren}. This implies that $M_{code}$ and $M_{phys}$ are isomorphic, i.e. $M_{code} = M_{phys}$, which means that the bulk and boundary relative entropies are equal.

\subsection{The Reeh--Schlieder theorem}
\label{sec:rs}

In the previous subsection, we explained how we use von Neumann algebras to approximately characterize bulk physics. Before we physically motivate the assumption in Theorem \ref{thm:maintheorem} that the set of cyclic and separating vectors with respect to $M_{code}$ is dense in $\calh_{code}$, we outline the conclusions of the Reeh--Schlieder theorem. Our discussion of the Reeh--Schlieder theorem follows the spirit of \cite{Witten:2018zxz}.

For the purposes of presenting the Reeh--Schlieder theorem, we restrict ourselves to quantum field theory in $d$-dimensional Minkowski space. Let $P^\mu$ be the energy-momentum operator. Each component of $P^\mu$ is a self-adjoint operator with its own set of spectral projections. Let $S_\Lambda$ be the subset of momentum space defined by 
$$S_\Lambda = \{p^\mu : |p^\mu| < \Lambda \quad \forall \mu \in \{ 0,1,\cdots, d-1 \} \}$$ 
for some cutoff energy $\Lambda$. Using the spectral projections of each $P^\mu$, we may construct a projection operator $\Pi_{S_\Lambda}$ that projects onto the subspace of states with energy-momentum in $S_\Lambda$. As $P^\mu$ is defined by smearing the local operator $T^{0\mu}$ (where $T^{\mu \nu}$ is the stress tensor) over an entire spatial slice,\footnote{Technically, a spatial slice is not an open subregion of spacetime.} $\Pi_{S_\Lambda}$ leaves each superselection sector invariant. Furthermore, for every $\ket{\Psi} \in \calh$,
$$ \lim_{\Lambda \rightarrow \infty} \Pi_{S_\Lambda} \ket{\Psi} = \ket{\Psi}.$$
Thus, the set of states of bounded energy-momentum in a given superselection sector is dense in that superselection sector.

The Reeh--Schlieder theorem may be applied to states of bounded energy-momentum. Let $\ket{\Xi}$ denote such a state. Let $\Sigma$ denote a spatial slice. Given an open proper subregion $\calv \subset \Sigma$, let $\calu_\calv$ be a small neighborhood in spacetime containing $\calv$. The Reeh--Schlieder theorem tells us that the closure of the set of states obtained by acting on $\ket{\Xi}$ with operators in the algebra $\cala(\calu_\calv)$ is equal to the closure of the set of states obtained by acting on $\ket{\Xi}$ with all local operators, which is the superselection sector of $\ket{\Xi}$.

Let us restrict our attention to a single superselection sector. Then $\ket{\Xi}$ is cyclic with respect to $\cala(\calu_\calv)$ and $M(\calu_\calv)$. Since $\calv$ is a proper subregion of $\Sigma$, the Reeh--Schlieder theorem may also be applied to the subregion $\calu_{\calv^\prime}$, where $\calv^\prime$ is the complement of the closure of $\calv$ in $\Sigma$. The result is that $\ket{\Xi}$ is also separating with respect to $M(\calu_\calv)$ \cite{Witten:2018zxz}. Thus, in quantum field theory in Minkowski space restricted to a single superselection sector, the fact that the set of states of bounded energy-momentum is dense implies that the set of cyclic and separating vectors with respect to $M(\calu_\calv)$ is dense.

\subsection{Physical motivation for the assumptions of Theorem \ref{thm:maintheorem}}
\label{sec:rsapplied}

We now use the Reeh--Schlieder theorem to understand the assumptions in Theorem \ref{thm:maintheorem} in a physical context. Without loss of generality, we assume that the bulk-to-boundary isometry $u$ in Theorem \ref{thm:maintheorem} maps $\calh_{code}$ into a single superselection sector of $\calh_{phys}$. That is, the code subspace lies within a single superselection sector. If this is not the case, then we can decompose $\calh_{code}$ into orthogonal subspaces that each are mapped into different superselection sectors of the boundary theory, and we can study Theorem \ref{thm:maintheorem} separately for each orthogonal subspace.

In Theorem \ref{thm:maintheorem}, we assume that the set of cyclic and separating vectors with respect to $M_{code}$ is dense in $\calh_{code}$. If the bulk theory was quantum field theory in Minkowski space, then the discussion in Section \ref{sec:rs} directly applies. However, the discussion in Section \ref{sec:rs} does not directly imply this because the bulk theory is only approximately described by quantum field theory and the background spacetime is asymptotically AdS. In \cite{Morrison}, a version of the Reeh--Schlieder theorem is proved for free scalar fields in global AdS. The theorem is valid for the vacuum state of the field quantized in global AdS, the vacuum state of the field quantized in any causal wedge, and finite-energy excitations of these vacua. If we choose to ignore the gravitational backreaction in the bulk and take $\calh_{code}$ to consist of finite-energy excitations of the global AdS vacuum, the results of \cite{Morrison} suggest to us that it is plausible that the set of cyclic and separating vectors with respect to $M_{code}$, where $M_{code}$ is associated with an entanglement wedge, is dense in the bulk vacuum superselection sector $\calh_0$. If $\calh_0$ is a proper subset of $\calh_{code}$, we should redefine $\calh_{code}$ to be $\calh_0$ for Theorem \ref{thm:maintheorem} to apply.

It would be interesting to investigate the plausibility of the assumption that the set of cyclic and separating states with respect to $M_{code}$ is dense in $\calh_{code}$ when $\calh_{code}$ contains black hole microstates. For a sufficiently large boundary subregion, the entanglement wedge of $M_{code}$ will contain the black hole, and the operators in $M_{code}$ correspond to local operators associated with the field degrees of freedom outside of the black hole as well as operators that act on the black hole microstates, whose description involves quantum gravity at the Planck scale. In quantum field theory, it is possible to generate the whole Hilbert space by acting on the vacuum with operators in a small subregion because the vacuum is highly entangled. It would be interesting to understand how the presence of a black hole changes the structure of spacetime entanglement outside the horizon. Holographic tensor network models suggest that entanglement wedge reconstruction is possible in the presence of a black hole \cite{Harlow:2018fse}; operators outside the black hole can in fact be ``pushed through'' the black hole tensor \cite{alphabits}. However, tensor network models of holography involve finite dimensional Hilbert spaces and thus cannot capture the pattern of entanglement that makes the Reeh--Schlieder theorem work.

%{\color{red} Insert paragraphs on Connes-Stormer result. -- might become a separate section or a subsection}

Finally, we address the assumption in Theorem \ref{thm:maintheorem} that for all states $\ket{\Psi} \in \calh_{code}$ that are cyclic and separating with respect to $M_{code}$, $u \ket{\Psi}$ is cyclic and separating with respect to $M_{phys}$. In \cite{Morrison}, the Reeh--Schlieder theorem holds for the vacuum of global AdS, implying that the vacuum is cyclic and separating with respect to the local operator algebra associated with a bulk subregion. The image of the bulk vacuum state under the bulk-to-boundary isometry is the boundary vacuum state, which is cyclic and separating with respect to the local operator algebras associated with boundary subregions. Likewise, finite-energy excited states in the bulk map to states in the boundary CFT of bounded energy-momentum, which are also cyclic and separating. This supports the assumption of Theorem \ref{thm:maintheorem} that the cyclic and separating states with respect to $M_{code}$ map to the cyclic and separating states with respect to $M_{phys}$.

\subsection{von Neumann algebra with type III$_1$ factors as a special case}
\label{sec:connesstormer}

Our main physical justification of the assumption that cyclic and separating states with respect to $M_{code}$ are dense in $\calh_{code}$ is the fact that the Reeh--Schlieder theorem applies to states of bounded energy-momentum, which are dense in the Hilbert space. In a generic local quantum field theory, the von Neumann algebra of a type III$_1$ factor\footnote{The definition of a type III$_1$ factor is given in \cite{TensorNetwork}.} is associated with a causal subregion of the spacetime. When $M_{code}$ and $M_{code}^\prime$ are type III$_1$ factors, the assumption of Theorem \ref{thm:maintheorem} that cyclic and separating states with respect to $M_{code}$ are dense in $\calh_{code}$ also follows from a result of Connes--St{\o}rmer, which is presented below. 
\begin{thm}[Connes--St{\o}rmer \cite{ConnesStormer}]
	A factor $M$ is of type III$_1$ if and only if the action of its unitary group on its state space by inner automorphisms is topologically transitive in the norm topology. 		
\end{thm}
Let $\ket{\Psi}$ be a cyclic and separating vector with respect to $M$. The above theorem implies that the set of vectors that can be written as $U U^\prime \ket{\Psi}$, where $U \in M$ and $U^\prime \in M^\prime$ are both unitary operators, is dense in $\calh$. Given that $\ket{\Psi}$ is cyclic and separating with respect to $M$, $U U^\prime \ket{\Psi}$ is also cyclic and separating. The existence of one cyclic and separating vector $\ket{\Psi}$ in Theorem \ref{thm:maintheorem} guarantees, for a factor of type III$_1$, that the set of cyclic and separating vectors with respect to $M_{code}$ is dense in $\calh_{code}$.

\subsection{Finite-dimensional quantum error correction}
\label{sec:implications}

In this section, we explain Theorem \ref{thm:maintheorem} in the context of previous work on finite-dimensional error correction \cite{Harlow:2016fse,Pastawski:2015qua,DongHarlowWall}. First, we interpret the assumption that cyclic and separating vectors with respect to $M_{code}$ map to cyclic and separating vectors with respect to $M_{phys}$ in the case that $\calh_{code}$ and $\calh_{phys}$ are finite dimensional. 
 As discussed in \cite{Harlow:2016fse}, a finite dimensional $M_{code}$ induces a decomposition of the code subspace,
\begin{equation}
\calh_{code} = \oplus_\alpha \calh_{a_\alpha} \otimes \calh_{\bar{a}_\alpha},
\end{equation}
such that any $\calo \in M_{code}$ may be written in block-diagonal form:
\begin{equation}
\calo = \left(\begin{array}{ccc}
\calo_{a_1} \otimes I_{\bar{a}_1} & 0 & \cdots \\ 
0 & \calo_{a_2} \otimes I_{\bar{a}_2} & \cdots \\ 
\vdots & \vdots & \ddots
\end{array} \right).
\label{eq:blockdiagonal}
\end{equation}
In the setup of \cite{Harlow:2016fse}, $\calh_{phys}$ may be written in the factorized form $\calh_{phys} = \calh_A \otimes \calh_{\bar{A}}$ where each factor corresponds to a boundary subregion and its complement. Let $M_{phys}$ induce the factorization $\calh_{phys} = \calh_{A} \otimes \calh_{\bar{A}}$ such that operators in $M_{phys}$ act trivially on $\calh_{\bar{A}}$. As \cite{Harlow:2016fse} points out, subalgebra codes with complementary recovery are especially relevant for AdS/CFT as they display a Ryu--Takayanagi formula with a nontrivial area operator. For such codes, an orthonormal basis of $\calh_{a_\alpha} \otimes \calh_{\bar{a}_\alpha}$ may be written as
\begin{equation}
\label{eq:finitedimcode}
u\ket{\alpha,ij}_{\text{code}} = U_A U_{\bar{A}} \left(\ket{\alpha,i}_{A_1^\alpha} \ket{\alpha,j}_{\bar{A}_1^\alpha} \ket{\chi_\alpha}_{A_2^\alpha \bar{A}_2^\alpha}\right),
\end{equation}
for a decomposition of $\calh_A$ given by 
\begin{align}
\calh_A = \oplus_\alpha (\calh_{A_1^\alpha} \otimes \calh_{A_2^\alpha}) \oplus \calh_{A_3},
\end{align}
and similarly for $\calh_{\bar{A}}$. Also, 
$$\dim \calh_{A_1^\alpha} = \dim \calh_{a_\alpha} \text{ and } \dim \calh_{\bar{A}_1^\alpha} = \dim \calh_{\bar{a}_\alpha}.$$
For each $\alpha$, $i$ and $j$ are indices that denote basis vectors in $\calh_{a_\alpha}$ and $\calh_{\bar{a}_\alpha}$ respectively. We have explicitly included $u$, the isometry from the code subspace to the physical Hilbert space. $U_A,U_{\bar{A}}$ are unitary matrices that act on $\calh_A,\calh_{\bar{A}}$, and $\ket{\chi}_{A_2^\alpha \bar{A}_2^\alpha}$ is a state that depends on the specific code under consideration. It is important that in the state $\ket{\chi}_{A_2^\alpha \bar{A}_2^\alpha}$, subsystems $A_2^\alpha$ and $\bar{A}_2^\alpha$ are entangled. If $\ket{\chi}_{A_2^\alpha \bar{A}_2^\alpha}$ were a factorized state for every $\alpha$, then it would not be possible to express $\ket{\alpha,ij}_{\text{code}}$ as in \eqref{eq:finitedimcode} for arbitrary choices of the factorization $\calh_{phys} = \calh_{A} \otimes \calh_{\bar{A}}$. That is, the code would not be useful for studying bulk reconstruction for arbitrary choices of boundary subregions. Furthermore, equation (5.26) of \cite{Harlow:2016fse} would imply that the area operator vanishes.

We now discuss the implications of Theorem \ref{thm:maintheorem} for the state $\ket{\chi_\alpha}_{A_2^\alpha \bar{A}_2^\alpha}$. Let us assume that $\dim \calh_A = \dim \calh_{\bar{A}}$ and that for every $\alpha$, $\dim \calh_{a_\alpha} = \dim \calh_{\bar{a}_\alpha}$. Otherwise, there do not exist any cyclic and separating vectors with respect to $M_{code}$ or $M_{phys}$. A vector in $\calh_{phys}$ is cyclic and separating with respect to $M_{phys}$ if and only if it has maximal Schmidt number with respect to $\calh_{phys} = \calh_{A} \otimes \calh_{\bar{A}}$. The assumption that cyclic and separating vectors with respect to $M_{code}$ map to cyclic and separating vectors with respect to $M_{phys}$ implies that  $\ket{\chi}_{A_2^\alpha \bar{A}_2^\alpha}$ must have maximal Schmidt number with respect to the factorization $\calh_{A_2^\alpha} \otimes \calh_{\bar{A}_2^\alpha}$ and that $\dim \calh_{A_2^\alpha} = \dim \calh_{\bar{A}_2^\alpha}$.  To see why, note that a cyclic and separating vector $\ket{\Phi} \in \calh_{code}$ with respect to $M_{code}$ may be written as \begin{equation}\label{eq:Phi}\ket{\Phi} = \sum_{\alpha,i,j} c^\alpha_{ij} \ket{\alpha,ij}_{\text{code}},\end{equation} where $c^\alpha_{ij}$ is a full-rank square matrix for each $\alpha$. Using equation \eqref{eq:finitedimcode} to map $\ket{\Phi}$ to $u\ket{\Phi} \in \calh_{phys}$, we see that if $\ket{\chi_{\hat{\alpha}}}_{A_2^{\hat{\alpha}} \bar{A}_2^{\hat{\alpha}}}$ does not have maximal Schmidt number for some $\hat{\alpha}$, then we can annihilate $u\ket{\Phi}$ with an operator that, up to conjugation by $U_A$, acts as the identity on $\calh_{\bar{A}}$, annihilates $\calh_{A_3}$, annihilates $\calh_{A_1^\alpha} \otimes \calh_{A_2^\alpha}$ for $\alpha \neq \hat{\alpha}$, and acts nontrivially on $\calh_{A_1^\alpha} \otimes \calh_{A_2^\alpha}$. This implies that $u\ket{\Phi}$ is not separating with respect to $M_{phys}$, which contradicts the assumption. Another consequence of the assumption is that $\calh_{A_3}$ and $\calh_{\bar{A}_3}$ must be trivial. Previous work on finite-dimensional error correction \cite{Harlow:2016fse,Pastawski:2015qua} has highlighted the crucial role of entanglement in bulk reconstruction. We have shown that the Reeh--Schlieder theorem suggests that cyclic and separating vectors with respect to $M_{code}$ are mapped via the bulk-to-boundary isometry to vectors that are cyclic and separating with respect to $M_{phys}$. In the context of finite-dimensional subalgebra codes, this implies that the area term in the Ryu--Takayangi formula cannot vanish. 

Our proof of entanglement wedge reconstruction in Theorem \ref{thm:maintheorem} is constructive. Given a bulk operator $\calo \in M_{code}$, equation \eqref{eq:reconstructionformula} provides an explicit formula for a boundary operator $\tilde{\calo} \in M_{phys}$. In order to understand our formula in the finite dimensional case, we use the decomposition $\calh_A = \oplus_\alpha (\calh_{A_1^\alpha} \otimes \calh_{A_2^\alpha})$ (and similarly for $\calh_{\bar{A}}$) and let $\ket{\Phi}$ (defined in equation \eqref{eq:Phi}) be our fiducial state. The action of $\calo \in M_{code}$ on a code subspace basis vector is
\begin{equation}
\calo \ket{\Phi} = \sum_{\alpha,i,\hat{i},j} c^\alpha_{ij} \braket{\hat{i}|\calo_{a_\alpha}|i} \ket{\alpha,\hat{i} j}_{\text{code}},
\end{equation}
where $\calo_{a_\alpha}$ is defined in equation \eqref{eq:blockdiagonal}. By equation \eqref{eq:finitedimcode} we then have that
\begin{equation}
u\ket{\Phi} = \sum_{\alpha,i,j} c^\alpha_{ij}
U_A U_{\bar{A}}\left(\ket{\alpha,i}_{A_1^\alpha} \ket{\alpha,j}_{\bar{A}_1^\alpha} \ket{\chi_\alpha}_{A_2^\alpha \bar{A}_2^\alpha}\right),
\end{equation}
\begin{equation}
u \calo \ket{\Phi} = \sum_{\alpha,i,\hat{i},j} c^\alpha_{ij} \braket{\hat{i}|\calo_{a_\alpha}|i}
U_A U_{\bar{A}} \left(\ket{\alpha,\hat{i}}_{A_1^\alpha} \ket{\alpha,j}_{\bar{A}_1^\alpha} \ket{\chi_\alpha}_{A_2^\alpha \bar{A}_2^\alpha}\right).
\end{equation}
Equation \eqref{eq:reconstructionformula} then defines $\tilde{\calo} \in M_{phys}$ by
\begin{align}
\begin{split}
& \tilde{\calo} \ \calp^\prime U_A U_{\bar{A}} 
\sum_{\alpha,i,j} c^\alpha_{ij}
\left(\ket{\alpha,i}_{A_1^\alpha} \ket{\alpha,j}_{\bar{A}_1^\alpha} \ket{\chi_\alpha}_{A_2^\alpha \bar{A}_2^\alpha}\right) \\
&\quad\quad\quad\quad\quad\quad := \sum_{\alpha,i,\hat{i},j} c^\alpha_{ij} \braket{\hat{i}|\calo_{a_\alpha}|i}
\calp^\prime U_A U_{\bar{A}} \left(\ket{\alpha,\hat{i}}_{A_1^\alpha} \ket{\alpha,j}_{\bar{A}_1^\alpha} \ket{\chi_\alpha}_{A_2^\alpha \bar{A}_2^\alpha}\right),
\end{split}
\end{align}
where $\calp^\prime \in M_{phys}$ can be any operator that acts as the identity on $\calh_A$. With a suitable choice of $\calp^\prime$, we may show that for any $\alpha,i,j$,
\begin{equation}
\tilde{\calo}\ U_A U_{\bar{A}} 
\left(\ket{\alpha,i}_{A_1^\alpha} \ket{\alpha,j}_{\bar{A}_1^\alpha} \ket{\chi_\alpha}_{A_2^\alpha \bar{A}_2^\alpha}\right) = 
 U_A U_{\bar{A}} \left(\sum_{\hat{i}}  \braket{\hat{i}|\calo_{a_\alpha}|i}\ket{\alpha,\hat{i}}_{A_1^\alpha} \ket{\alpha,j}_{\bar{A}_1^\alpha} \ket{\chi_\alpha}_{A_2^\alpha \bar{A}_2^\alpha}\right).
\end{equation}
Thus, Theorem \ref{thm:maintheorem} along with the reconstruction formula in equation \eqref{eq:reconstructionformula} is an appropriate infinite-dimensional generalization of the finite-dimensional subalgebra codes with complementary recovery studied in \cite{Harlow:2016fse}.

\subsection{Outlook for holographic relative entropy}

The entanglement wedge reconstruction proposal is an example of bulk reconstruction. It asserts that for holographic theories, local operators in the entanglement wedge of a boundary subregion $A$ can be written in terms of CFT operators localized on $A$ \cite{Harlow:2018fse,DongHarlowWall,Jafferis:2015del}. Assuming that the operators in $M_{code}$ and $M_{code}^\prime$ in Theorem \ref{thm:maintheorem} lie respectively in an entanglement wedge and its complement, Theorem \ref{thm:maintheorem} establishes entanglement wedge reconstruction from the equivalence of bulk and boundary relative entropies and vice versa. Thus, it has been suggested that the entanglement wedge is ``dual'' to its corresponding boundary subregion \cite{Jafferis:2015del}. Another interesting result of \cite{Jafferis:2015del} is that bulk modular flow is dual to boundary modular flow, which we have captured in equation \eqref{eq:modularoperators}. The bulk and boundary modular operators act on the code subspace in the same way.

Quantum error correction in finite dimensional Hilbert spaces has been crucially used to argue for the entanglement wedge reconstruction proposal \cite{DongHarlowWall,Harlow:2016fse}. When $\calh_{code}$ and $\calh_{phys}$ are finite-dimensional, Theorem \ref{thm:maintheorem} has parallels to Theorem 1.1 of \cite{Harlow:2016fse}. In Theorem \ref{thm:maintheorem}, we assume that cyclic and separating vectors with respect to $M_{code}$ are dense in $\calh_{code}$, which is essentially a bulk version of the Reeh--Schlieder Theorem \cite{Morrison}. We also assume that cyclic and separating states with respect to $M_{code}$ map to cyclic and separating states with respect to $M_{phys}$, the algebra corresponding to a boundary subregion. These assumptions guarantee that the subalgebra codes studied in \cite{Harlow:2016fse} have a nonzero area operator. \cite{Harlow:2016fse} defines relative entropy in the boundary theory as $S(\rho,\sigma) = \text{Tr } \rho(\log \rho - \log \sigma)$. The definition of relative entropy we use in the bulk and boundary is appropriate for infinite-dimensional Hilbert spaces and reduces to the aforementioned formula in the finite-dimensional case \cite{Araki}. Thus, we have shown that the relative entropy formula in \cite{Araki} naturally describes the holographic relative entropy in quantum field theory to order $G_N^0$.

\section*{Acknowledgments}
The authors are grateful to Daniel Harlow, Temple He, Sungkyung Kang, and Kai Xu for discussions. M.J.K. would like to thank Simons workshop 2018, Strings 2018, and String-math 2018, Virginia Tech, University of Pennsylvania, Simons Center of Geometry and Physics, and Caltech for their hospitality. D.K. would like to thank PITP 2018. M.J.K. and D.K. would like to acknowledge a partial support from NSF grant PHY-1352084.

\end{document}